\documentclass[11pt, epsfig]{article}
\usepackage{epsfig, amsmath, amssymb, amsthm, times}
\usepackage{graphicx}

\usepackage[mathcal]{eucal}
\usepackage{bbm}

\newtheorem {thm}{Theorem}[section]
\newtheorem {lem}[thm]{Lemma}

\newtheorem {cor}[thm]{Corollary}

\theoremstyle{defintion}

\theoremstyle{remark}

\theoremstyle{example}

\theoremstyle{assumption}

\def\E{{\mathbb E~}}
\def\P{{\mathbb P~}}
\def\R{{\mathbb R}}

\def\SS{{\mathbb S}}

\def\lbl{\label}
\def\be{\begin{equation}}
\def\ee{\end{equation}}
\def\p{\partial}

\def\deg{\operatorname{deg}}

\def\t{\mathsf{T}}

\def\1{\mathbf{1}}
 \newcommand{\iu}{{i\mkern1mu}}

\title{The Kuramoto model on power law graphs}
\author{ Georgi S. Medvedev\thanks{
 Department of Mathematics, Drexel University, 3141 Chestnut Street,
 Philadelphia, PA 19104;
 {\tt medvedev@drexel.edu}
 }\quad
 and \quad Xuezhi Tang\thanks{ Wells Fargo Securities, 505 S. Tryon street, Charlotte, NC 28202;
{\tt xuezhi.tang@wellsfargo.com}
} 
}

\begin{document}
\maketitle

\begin{abstract}
The Kuramoto model (KM) of coupled phase oscillators on scale free graphs is analyzed in this work. 
The W-random graph model is used to define a convergent family of sparse graphs with power 
law degree distribution. For the KM on this family of graphs, we derive the mean field description
of the system's dynamics in the limit as the size of the network tends to infinity. The mean field 
equation is used to study two problems: synchronization in the coupled system with randomly
 distributed 
intrinsic frequencies and existence and bifurcations of chimera states in the KM with repulsive
coupling. The analysis of both problems highlights the role of the scale free network 
organization 
in shaping dynamics of the coupled system. The analytical results are complemented
 with the 
results of numerical simulations.
\end{abstract}

\section{Introduction}
\setcounter{equation}{0}

Coupled dynamical systems on graphs serve as mathematical models of various 
technological physical, biological, social, and economic networks \cite{PG16}. Examples include
neuronal and genetic networks 
and models of flocking in life sciences \cite{MotTad14}; power and information networks
and consensus protocols in technology \cite{Med12}; and economic and social 
networks and models of opinion dynamics in social sciences \cite{PG16}. This list can be
continued. Numerical simulations and mathematical analysis of coupled systems provided
many important insights into the mechanisms underlying collective dynamics in complex
networks. In the last two decades, there have been a remarkable progress in understanding 
classical phenomena
such as synchronization and phase locking in complex networks 
\cite{ChiNis2011, Str00, WSG06, MedTan15a, MedWri17}, and the discoveries of
new effects in the dynamics of networks such as chimera states 
\cite{KurBat02, AbrStr06, Ome13}.
The research on dynamical networks has been fueled by the desire for
better understanding the link between the structure of a network and its dynamics. 
This is the main motivation of our work.

Real world networks feature a rich variety of connectivity patterns. Scale free 
networks have been singled out in the network science community for their nontrivial structure
and compelling applications. The latter include the world wide web and scientific citation  
network among other physical, biological, and social networks \cite{BAR99}. 
Scale free graphs are 
characterized
by power law asymptotics of the degree distribution. For this reason, they are also called
power law graphs. In practice, power law distribution is determined by statistical methods. 
Different combinatorial  algorithms such as the preferential attachment  (see, e.g., \cite{BCLV11}) 
and Chung-Lu \cite{ChuLu02} methods are used
to generate computational models of scale free graphs.
 Dynamical systems on the graphs generated by 
these methods are difficult to study analytically.
Consequently, there are few mathematical results on the dynamics of coupled 
systems on scale free graphs. The goal of this paper is to rectify this situation.
We introduce  a framework for
modeling and analyzing coupled systems on scale free graphs. 
Further, we illustrate the role 
that the scale free connectivity can play in shaping dynamics of coupled systems. 
To this end, 
we analyze two problems:  synchronization in attractively coupled KM on power 
law graphs with
randomly distributed intrinsic frequencies 
and chimera states in repulsively coupled KM. 

The KM of coupled phase oscillators is one of the most successful mathematical models for studying
collective dynamics and synchronization \cite{Kur75}. It captures the essential features of dynamics
of weakly coupled limit cycle oscillators \cite{HopIzh97} and has many interesting applications 
in physical and biological sciences \cite{Str00}. In the synchronization problem, the intrinsic frequencies of the 
individual oscillators are assumed to be taken from a probability distribution with density function 
$g$. Then one
wants to find a critical value of the coupling strength, which marks the transition from stochastic 
distribution of the phases of oscillators to synchronization. For the classical KM on complete graphs,
Kuramoto found the critical value $K_c=2\pi/g(0).$ Kuramoto's self-consistent 
analysis recently  received a rigorous mathematical
justification in  \cite{ChiNis2011, Chi15a}. In the present paper, we use the results in 
\cite{ChiMed16}, which
extend some of the techniques from \cite{ChiNis2011} to the KM on graphs. 
Using these techniques, we obtain an explicit 
formula for the onset of synchronization for scale free graphs.

The second problem considered in this paper deals  with chimera states. These are special spatio-temporal
patterns in dynamical networks,
which feature coexisting regions of synchronous (regular)
and stochastic behaviors. Chimera states were discovered by Kuramoto and Battogtokh
in the KM with random initial conditions  \cite{KurBat02, AbrStr06}.  Since then they have been studied
in different settings. Importantly, chimera states have 
been confirmed experimentally (see, e.g., \cite{NTS16}). 
In the last decade, chimera states have become a subject of intense research \cite{PanAbr15}.
In this paper, we  present a new simple mechanism for chimera states in the KM with 
repulsive coupling. Our mechanism exploits the scale free structure of the network and admits 
a simple and explicit analytical description.

The organization of the paper is as follows. In the next section, we explain the W-random graph
model of power law graphs, which will be used in this work. 
There we will also formulate the KM on the power law graphs following
\cite{KVMed17}. In Section~\ref{sec.mean},
we review the mathematical background of the mean field equation for the KM on graphs 
following \cite{ChiMed16}. Although, the results in \cite{ChiMed16} do not apply to the model
on the power law graphs due to the singularity in this graph model, they can be used for the 
truncated power law model. The latter captures all essential features of the original model
Section~\ref{sec.sync} deals with the synchronization problem for the KM on power law graphs,
and Section~\ref{sec.repulsive} -- with the chimera states in the repulsively coupled model.
We conclude with brief discussion in Section~\ref{sec.discuss}.

\section{The model}
\lbl{sec.model}
\setcounter{equation}{0}

W-random graphs provide a convenient framework for deriving the continuum limit of the KM
on convergent families of graphs \cite{Med14b, Med14c, KVMed17, ChiMed16}.
In this section, we explain the W-random graph model adapted from \cite{BorCha14}
that will be used below.

Let $W(x,y)=(xy)^{-\alpha}$ for some $\alpha\in (0,1)$,
\be \lbl{Xn}
X_n = \{ x_{n0}, x_{n1},x_{n2},\dots,x_{nn}\}, \; x_{ni}=i/n,\; i=0,1,\dots, n,
\ee
and $\rho_n=n^{-\beta},\; \alpha<\beta<1.$ 

$\Gamma_n=G(W,\rho_n,X_n)$ stands for a random graph with the node
set $V(\Gamma_n)=[n]$ and the edge set $E(\Gamma_n)$ 
defined as follows.
The probability that $\{i, j\}$ forms an edge is
\be\lbl{Pedge}
\P(\{i,j\} \in E(\Gamma_n))=\rho_n \bar W_n(x_{ni},x_{nj}), \; i,j\in [n],
\ee
where
\be
\lbl{Wn}
\bar W_n(x,y)= \rho_n^{-1}\wedge W(x,y).\footnotemark
\ee
\footnotetext{Throughout this paper,
  we use $a\wedge b$ and $a\vee b$ to denote $\min\{a,b\}$ and
  $\max\{a,b\}$ respectively.}
The decision whether a given pair of nodes is included in the edge set
is made independently from other pairs. In other words,
$G(W,\rho_n,X_n)$ 
is a product probability
space
\be\lbl{pspace}
(\Omega_n= \{0,1\}^{n(n+1)/2}, 2^{\Omega_n},\P).
\ee
By $\Gamma_n(\mathbf{\omega}), \omega\in \Omega_n$, we will denote
a random graph drawn from the probability distribution $G(W,\rho_n,X_n)$.

\begin{figure}
\begin{center}
\includegraphics[height=2.0in, width=2.75in]{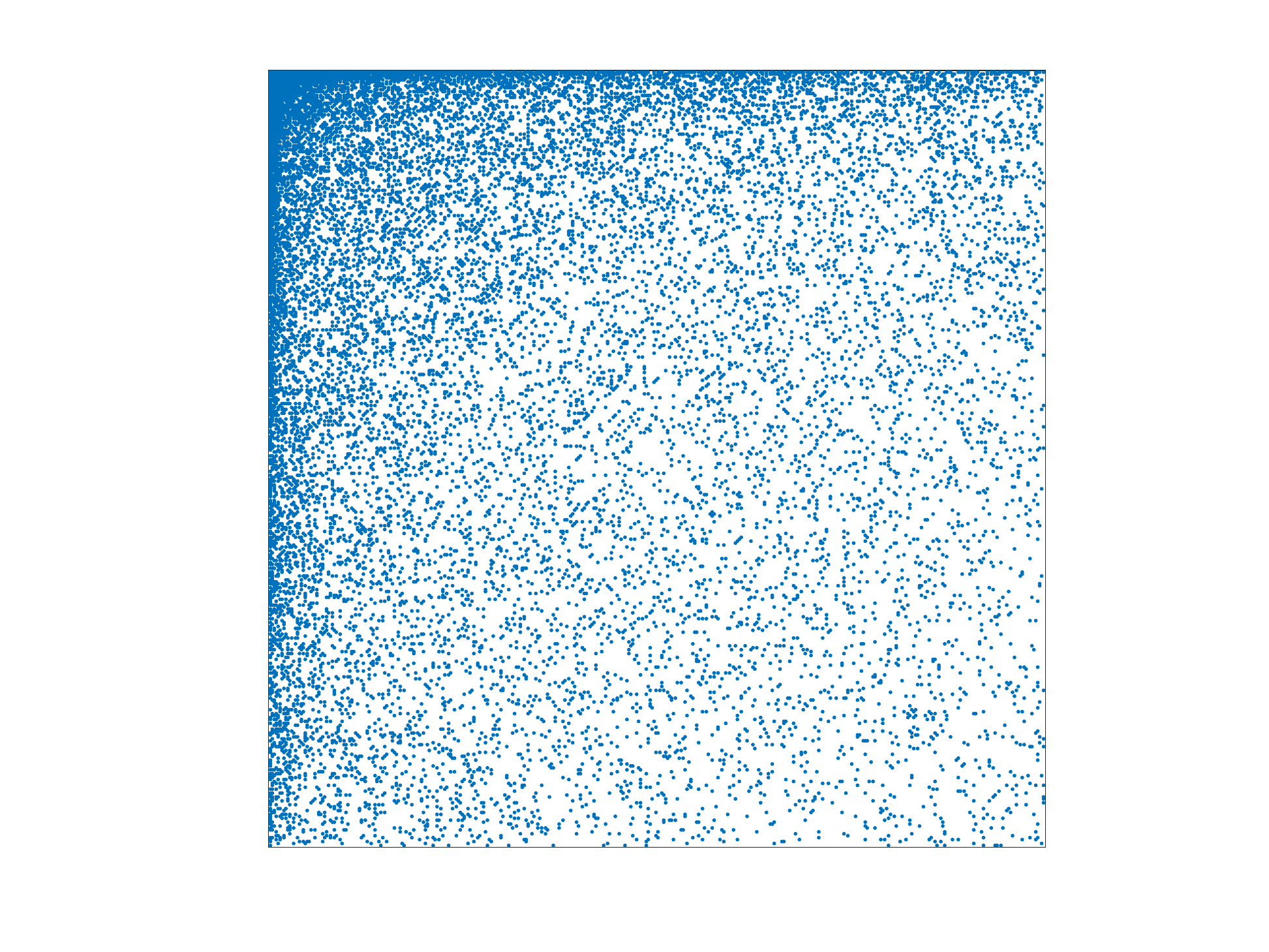}
\end{center}
\caption{The pixel plot of the adjacency matrix of a power law graph.
}
\lbl{f.-1}
\end{figure}

\begin{lem}\lbl{lem.pl}\cite{KVMed17} 
$\Gamma_n=G(W,\rho_n,X_n)$ has the following properties:
\begin{description}
\item[A)] The expected degree of node $i\in [n]$ of $\Gamma_n$ is\footnote{Here and below,
$\E_\omega$ denotes the mathematical expectation with respect to the probability space
\eqref{pspace}
underlying the random graph model.}
\be\lbl{Edeg}
\E_\omega\deg_{\Gamma_n}(i)=(1-\alpha)^{-1}n^{1+\alpha-\beta} i^{-\alpha}(1+o(1)).
\ee
\item[B)] The expected edge density of $\Gamma_n$ is $(1-\alpha)^{-2}n^{-\beta}(1+o(1)).$
\end{description}
\end{lem}

Let $\Gamma_n=\Gamma_n(\omega), \omega\in\Omega_n,$ be a random graph model
taken from  the probability distribution $G(W,\rho_n,X_n)$. The KM on $\Gamma_n$ is 
defined as follows:
\be\lbl{KM}
\dot u_{ni}= \phi_i+
{K\over n\rho_n} \sum_{ j=1}^n \xi_{nij}(\omega) \sin(u_{nj}-u_{ni}),\quad i\in[n],
\ee
where $\phi_i$ is the intrinsic frequency of the oscillator~$i$, 
$\xi_{nij}(\omega)=\1_{E(\Gamma_n(\omega))} (\{i,j\}) $
is a Bernoulli random variable, which takes value $1$ when $\{i,j\}$ is an edge of $\Gamma_n$.
$K$ controls the strength of coupling. The scaling factor $\rho_n$ on the right hand side of
\eqref{KM} is used to guarantee that the \eqref{KM} has a nontrivial continuum limit as $n\to\infty$.
In the classical KM on complete graph $\rho_n=1$. The same holds for the KM on any dense graphs
(i.e., the graph for which the edge density remains $O(1)$ for $n\gg 1$). However, for sparse 
graphs
like the power law graph $\Gamma_n=G( W, \rho_n, X_n)$, the edge density vanishes as
$n\to\infty$ (cf.~Lemma~\ref{lem.pl}). Thus, the scaling of the right hand side of \eqref{KM} 
by $\rho_n$ is needed, if one wants to have a nondegenerate continuum limit as $n\to\infty$. 

The KM \eqref{KM} is derived from a system of weakly coupled oscillators \cite{Kur75, HopIzh97}. 
The dynamical variable
$u_{ni}$ stands for the phase of oscillator $i$ and $\phi_i$ is its intrinsic frequency.
Until Section~\ref{sec.repulsive}, we assume that the intrinsic frequencies are drawn form a 
continuous probability
distribution with density $g$. 
The right hand side of \eqref{KM} depends on the realization of the random graph model 
$\Gamma_n=G(W, \rho_n, X_n)$. Thus, we are dealing with a system of 
ordinary differential equations with random coefficients. As a first step in analyzing \eqref{KM},
we substitute \eqref{KM} by the averaged model, which approximates the random KM \eqref{KM}.
Specifically,  we average the 
right-hand side of \eqref{KM} over all possible realizations of $\Gamma_n$:
\be\lbl{aKM}
 \dot v_{ni}(t)= F_{ni}(v_n), \; v_n(t)=(v_{n1}(t),v_{n2}(t),\dots, v_{nn}(t)),
\ee
where
\begin{equation*}
\begin{split}
F_{ni} (v) &= \E_\omega \left\{ \phi_i +K(n\rho_n)^{-1}
\sum_{j=1}^n \xi_{nij}(\omega) \sin(u_{nj}-u_{ni})\right\}\\
&=\phi_i +K(n\rho_n)^{-1}\sum_{j=1}^n \E_\omega\left(\xi_{nij}(\omega)\right) \sin(u_{nj}-u_{ni})\\
&=\phi_i +Kn^{-1}\sum_{j=1}^n \bar W_{nij} \sin(u_{nj}-u_{ni}),
\end{split}
\end{equation*}
where we used \eqref{Pedge}. Recall that $\E_\omega$ stands for the mathematical expectation
with respect to the probability space \eqref{pspace}.

Thus, the averaged  KM has the following form
\be\lbl{AKM}
\dot v_{ni}=\phi_i +Kn^{-1}\sum_{j=1}^n \bar W_{nij} \sin(v_{nj}-v_{ni}).
\ee
It is shown in \cite{ChiMed16} (see also \cite{KVMed17}) that for $n\gg 1,$ 
solutions of the initial value problems (IVPs)
 for the original and averaged
KMs subject to the same initial conditions remain close on finite time intervals with 
probability $1$. 
\begin{lem}\lbl{lem.average}\cite{ChiMed16}
\be\lbl{in-prob}
\lim_{n\to\infty} \max_{t\in [0,T]}
\left\|v_n(t)-u_n(t)\right\|_{2,n} =0 \quad \mbox{almost surely}.
\ee
\end{lem}
Here, 
\be\lbl{2n-norm}
\| v_n\|_{2,n}:=\sqrt{ n^{-1} \sum_{i=1}^n v_{ni}^2}
\ee
is the discrete $L^2$-norm.

\section{The mean field limit} \lbl{sec.mean}
\setcounter{equation}{0}

Our main tool for studying the KM in this paper is
the mean field equation. It is  derived in the limit as the number of oscillators tends to infinity.
In this section, we explain the mathematical basis of the mean field limit for the model
at hand.

In studies of large groups of interacting dynamical systems, it is often useful to take 
the limit as the size of system goes to infinity \cite{Gol16}. The continuum limit has been instrumental
in the analysis of synchronization and chimera states in the KM \cite{Str00, Ome13}.
For the discrete model \eqref{AKM}, the continuum limit has the following form 
\be\lbl{MF}
{\p \over \p t} \rho (t,u,\phi,x) +
{\p \over \p u } \left\{\rho(t, u,\phi,x) V(t,u,\phi,x) \right\}=0
\ee
where
\be\lbl{def-V}
V(t, u,\phi,x)=\phi+
K \int_I \int_\R \int_{\SS}  W(x,y) \sin(v - u ) \rho (t,v \phi, y) dv d\phi d y
\ee
Here, $\rho(t,u,\phi,x)$ stands for a probability density function on $G=\SS\times \R\times I$
parametrized by time $t\in\R^+$. It aims to describe the distribution of the oscillators of the discrete
model \eqref{KM} at time $t$, provided both IVPs for \eqref{KM} and \eqref{MF} are initialized 
appropriately.

For the original KM on complete graphs, the rigorous mathematical justification of the 
mean field equation \eqref{MF} was given by Lancellotti  \cite{Lan05}. It relies on the classical
theory for the Valsov equation \cite{Neu78}. For the KM on graphs, the justification of the mean field 
limit was developed in \cite{ChiMed16}.  The analysis in \cite{ChiMed16} uses Lipschitz
continuity of $W$ and does not apply to the singular kernel $W(x,y)=(xy)^{-\alpha}$ used in this 
paper. However, if $W$ is replaced by the truncated kernel 
$W_C(x,y)=C\wedge (xy)^{-\alpha}$ for arbitrary  $C>0$ then the interpretation of the mean field limit in
\cite{ChiMed16} carries over to the problem at hand. Using the truncated kernel does not limit 
the applications of our results, as all effects considered in this work can be achieved with $W_C$
instead of $W$. Thus, in the remainder of this section, we explain the mathematical meaning
of the mean field equation \eqref{MF} and its relation to the discrete system \eqref{AKM} assuming
that $W:=W_C$ for sufficiently large $C>0$.

Consider the following initial condition for \eqref{MF}
\be\lbl{MF-ic}
\rho(0,u,\phi,x)=\rho^0(u,\phi,x)g(\phi),
\ee
where the nonnegative $\rho^0\in L^1(G)$ satisfies
\be\lbl{MF-norm}
\int_\SS \rho^0(u,\phi,x)du=1\quad  \forall (\phi,x)\in \R\times I.
\ee
Then, as shown in \cite{ChiMed16}, there is a unique weak solution of the IVP \eqref{MF},
\eqref{MF-ic}. Moreover, $\rho(t,\cdot)$ is a probability density function on $G$
for every $t\in [0,T]$. Thus, one can define the probability measure
\be\lbl{CMeasure}
\mu_t(A)=\int_A \rho(t,u,\phi, x) du d\phi dx, \quad A\in\mathcal{B}(G),
\ee
where $\mathcal{B}(G)$ stands for the collection of Borel subsets of $G$.

On the other hand, the solution of the IVP for
discrete problem 
\eqref{AKM} defines the 
empirical measure
\be\lbl{EMeasure}
\mu_t^n(A)=n^{-1} \sum_{i=1}^n \1_A \left(\theta_{ni}(t), \phi_i, x_{ni}\right),\quad A\in\mathcal{B}(G).
\ee
The analysis in \cite{ChiMed16}, based on the Neunzert's theory for Vlasov equation
(cf.~\cite{Neu78}), shows that
\be\lbl{convergence-of-measures}
\sup_{t\in [0,T]} d(\mu_t^n, \mu_t)\to 0\quad \mbox{as}\quad n\to\infty,
\ee
provided
\be\lbl{convergence-of-measures-ic}
d(\mu_0^n, \mu_0)\to 0\quad \mbox{as}\quad n\to\infty,
\ee
Here, $d(\cdot,\cdot)$ stands for the bounded Lipschitz distance, which metrizes
weak convergence for the space of Borel probability measures on $G$ \cite{Dud02}.
Thus, if the initial distribution of oscillators (i.e., the initial conditions for \eqref{AKM})
converges weakly to $\mu_0$ as $n\to\infty$), then the solution of the continuous problem 
\eqref{MF}, \eqref{MF-ic} approximates the distribution of oscillators
around $\SS$ for every
$t\in [0,T]$. The same applies to the empirical measures generated by the  discrete 
model on random graph \eqref{KM}  via Lemma~\ref{lem.average} (cf.~\cite{ChiMed16}).

\section{Synchronization}\lbl{sec.sync}
\setcounter{equation}{0}
With the mean field limit \eqref{MF} in hand, we are equipped to study dynamics of the 
KM on power law graphs \eqref{KM}. First, note that $\rho(t,u,\phi,x)=(2\pi)^{-1} g(\omega)$ is a probability
density on $G$, corresponding to the state of the network with 
all oscillators distributed uniformly around $\SS$,  which will be referred to as the incoherent state. 
It is a solution of \eqref{MF}. 

\begin{figure}
\begin{center}
\textbf{a}\includegraphics[height=1.8in,width=2.0in]{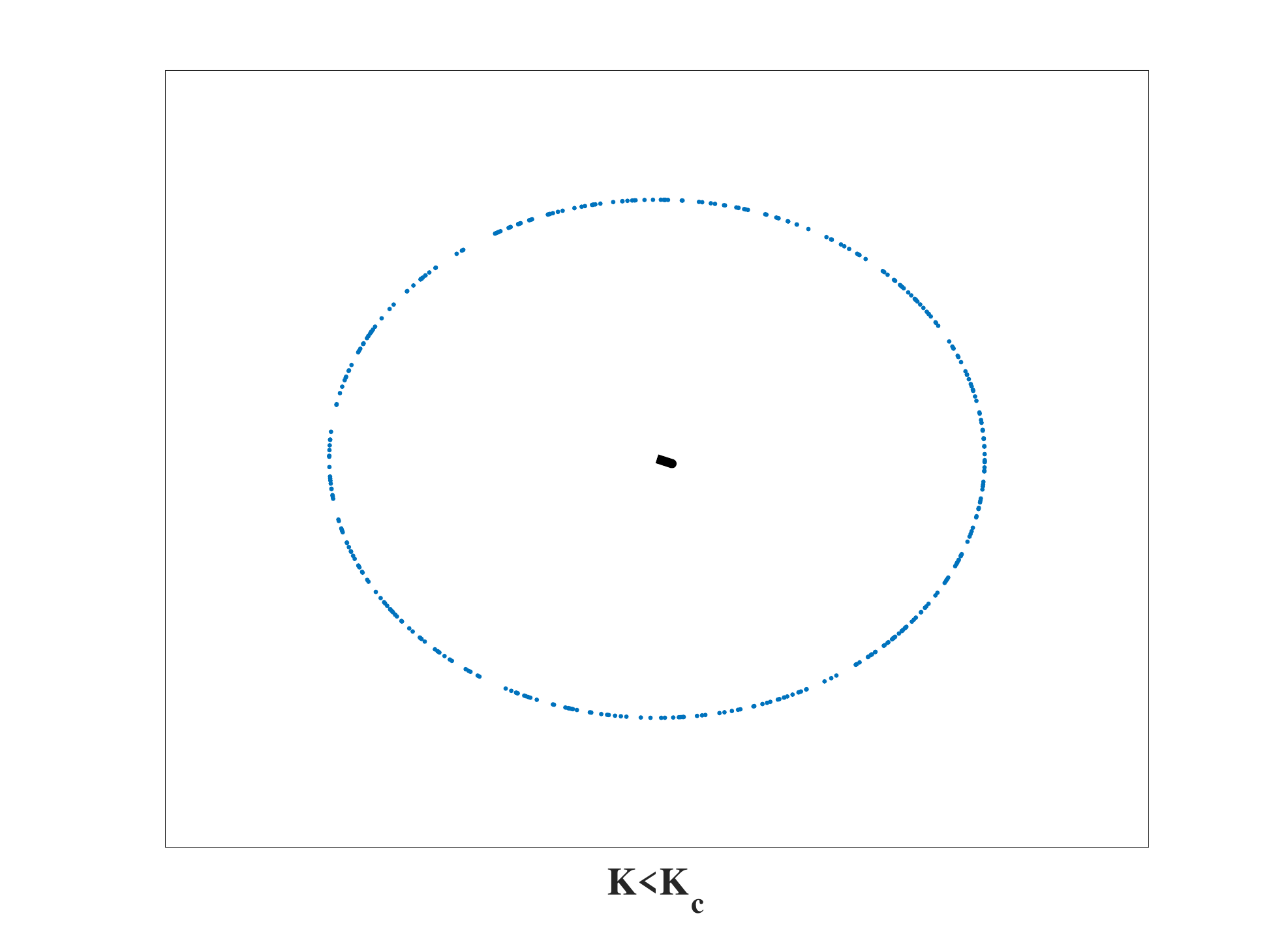}\qquad \qquad
\textbf{b}\includegraphics[height=1.8in,width=2.0in]{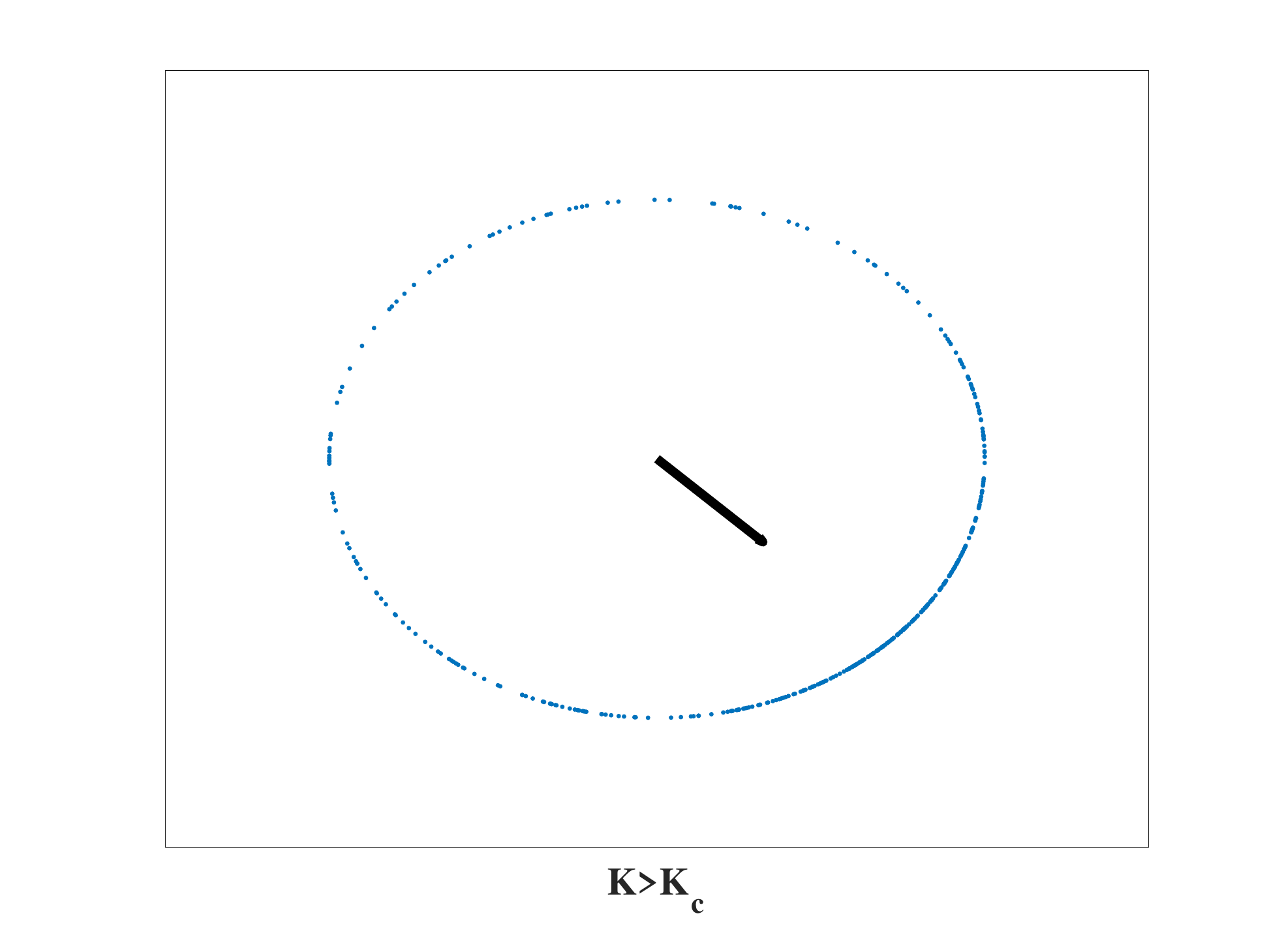}
\end{center}
\caption{The distribution of the oscillators for the values of the 
coupling strength below (\textbf{a})  and above (\textbf{b}) the synchronization 
threshold. The black arrows indicate the order parameter in each case.
}
\lbl{f.0}
\end{figure}

To study stability of the incoherent state, we use the results of \cite{ChiMed16}. To this 
end, throughout this section, we assume that $\alpha\in (0,1/2)$ and 
the probability density function $g$ characterizing the distribution of intrinsic
frequency $\phi$ is even. 

Next, consider 
a  self-adjoint operator $\mathcal{W}: L^2(I)\to L^2(I)$ defined by
\be\lbl{kernel}
\mathcal{W}[f]=\int_I W(\cdot, y)f(y)dy, \quad f\in L^2(I).
\ee

We will need the spectral properties of $\mathcal{W}$ summarized in the 
following lemma.
\begin{lem}\lbl{lem.EV}
For $\alpha\in (0,1/2)$, the spectrum of $\mathcal{W}$ consists of a simple 
eigenvalue $(1-2\alpha)^{-1}$ and the zero eigenvalue of infinite multiplicity.
\end{lem}
\begin{proof}
Suppose $\zeta$ is an EV of $\mathcal{W}$, and $f\in L^2(I)$ is a corresponding 
eigenfunction. Then
$$
x^{-\alpha}\int_I y^{-\alpha} f(y) dy =\zeta f(x).
$$
If $f$ is orthogonal to the subspace of $L^2(I)$ spanned by $x^{-\alpha}$, then 
$\zeta=0$. Otherwise, $f$ must be equal to $Cx^{-\alpha}$ for some $C\neq 0$
and
$$
\zeta=\int_I y^{-2\alpha}dy=(1-2\alpha)^{-1}.
$$
\end{proof}

By \cite[Theorem~3.4]{ChiMed16}, the incoherent state $\rho=(2\pi)^{-1} g(\omega)$ 
is  linearly stable for $K\le K_c,$ where the critical coupling strength 
\be\lbl{Kcritical}
K_c={2\over \pi g(0) \zeta_{max}(\mathcal{W})},
\ee
and $\zeta_{max}(\mathcal{W})$ is the largest positive eigenvalue of  $\mathcal{W}$. 
Using \eqref{Kcritical} and Lemma~\ref{lem.EV}, we conclude that the onset of synchronization 
for the KM on scale free graphs is defined by the 
critical value 
\be\lbl{critical-pwl}
K_c={2(1-2\alpha)\over \pi g(0)}, \quad \alpha\in (0,1/2).
\ee
In particular, for the Gaussian density $g(\phi)= e^{-\phi^2/2}/\sqrt{2\pi}$, we have
$$
K_c(\alpha)=2\sqrt{2\over \pi}(1-2\alpha).
$$
For $K\le K_c$, the incoherent state is linearly stable and the oscillators are distributed 
approximately uniformly around $\SS$ (see Fig.~\ref{f.0}a). For values of $K>K_c$, numerics
shows a gradual buildup of coherence in the system dynamics (see Fig.~\ref{f.0}b).
Note that the synchronization threshold \eqref{critical-pwl} can be made arbitrarily small by 
taking $\alpha$
close $1/2$. Such good synchronizability of the network is an implication of the scale free
connectivity.

\section{Repulsive coupling}\lbl{sec.repulsive}
\setcounter{equation}{0}

In the previous section, we found that the incoherent state is stable for $K\le K_c$ and
is unstable otherwise. For $K>K_c,$ coherence gradually builds up and the asymptotic 
state of the system becomes closer and closer to complete synchronization as $K\to\infty$.
In this section, we focus on pattern formation in repulsively coupled networks, i.e., 
we consider \eqref{KM} for $K<0$. In this case, the incoherent state is stable, but as we will
see below, there are other stable states. To make the model analytically tractable,
we set the intrinsic 
frequencies equal to the same value $\phi_i=\phi,\; i\in [n]$. By switching to a uniformly
rotating frame of coordinates, without loss of generality we assume $\phi=0.$ 

Thus, in the remainder of this section, we will be dealing with the following model
\be\lbl{rKM}
\dot u_{ni}=(n\rho_n)^{-1}\sum_{j=1}\xi_{nij}(\omega) \sin (u_{ni}-u_{nj}),\quad i\in [n].
\ee
Here,  we set $K$ to $-1$, since this can always be achieved by rescaling time.

The averaged model then becomes
\be\lbl{rAKM}
\dot v_{ni}=n^{-1}\sum_{j=1} \bar W_{nij} \sin (v_{ni}-v_{nj}),\quad i\in [n].
\ee
The corresponding mean field equation is given by
\be\lbl{rMF}
{\p \over \p t} \rho(t,u,x) +{\p\over \p u}\left\{ V(t,u,x) \rho(t,u,x)\right\}=0,
\ee
where
\be\lbl{rV}
V(t,u,x)= \int_\SS \int_I W(x,y) \sin(u-v) \rho(t,v,y) dvdy.
\ee

The synchronous state is always unstable
for the KM with repulsive coupling  on any undirected graph 
(cf.~\cite[Theorem~3.7]{MedTan15a}). Thus, it is unstable for the 
KM on power law graphs \eqref{rKM}. The proof of this fact uses a 
variational argument. A similar variational principle applies to the 
averaged model \eqref{rAKM}. It is used in the next subsection to 
identify attractors in the repulsively coupled model.

\subsection{The Lyapunov function}
For the classical KM on complete graphs, Kuramoto introduced the order 
parameter
\be\lbl{Korder}
R_{cmp}(u)=n^{-1} \sum_{j=1}^n e^{\iu u_i}, \quad u=(u_1,u_2, \dots, u_n)\in \SS^n.
\ee
to study the transition to synchronization \cite{Kur75}. The complex-valued order
parameter \eqref{Korder} provides a convenient measure of coherence in the system dynamics. 
Indeed, if the phases $u_i, \; i\in [n],$ are distributed around $\SS$ uniformly  then $|R_{cmp}(u)|\approx 0$,
whereas if they evolve in synchrony then $|R_{cmp}(u)|\approx 1$. 

For the KM on weighted 
graphs \eqref{rAKM}, there is a  suitable generalization of the order parameter:
\be\lbl{order}
R(u)=n^{-1} \sum_{j=1}^n a_{nj} e^{\iu u_i}, \quad a_{ni}=x_{ni}^{-\alpha}, \; i\in [n].
\ee
Using $R$, \eqref{rAKM} can be rewritten as follows
\be\lbl{RAKM}
\dot u_{ni} = a_{ni} |R(u_n)| \sin ( u_{ni}-\psi), \; i\in [n],
\ee
where $\psi=\operatorname{Arg} R(u_n)$. As follows from \eqref{RAKM}, there are
two classes of equilibria of \eqref{AKM}:
\begin{eqnarray*}
\mathcal{E}_{n,1} &=& \left\{ u\in \SS^n:\; \left(R(u)\neq 0\right)\; \&\; \left(u_j-u_i\in \{0,\pi\}, \; 
i,j\in [n]\right)\right\},\\
\mathcal{E}_{n,2} &=& \left\{ u\in \SS^n:\; R(u)= 0 \right\}.
\end{eqnarray*}

\begin{thm}\lbl{thm.omega}
The $\omega$-limit set of \eqref{AKM} is $\mathcal{E}_{n,1}\cup \mathcal{E}_{n,2}$.
\end{thm}
\begin{proof}
Let 
\be\lbl{def-L}
L(u)={1\over 2n} | R(u)|^2
\ee
and note that 
\be\lbl{L-note}
\begin{split}
2nL(u) &= \left[ n^{-1} \sum_{j=1}^n a_{nj} \cos u_{ni}\right]^2 +
\left[ n^{-1} \sum_{j=1}^n a_{nj} \sin u_{ni}\right]^2\\
&= n^{-2} \sum_{i,j=1}^n a_{ni} a_{nj} \cos(u_{ni}-u_{nj})
\end{split}
\ee
Further, 
\be\lbl{dLdu}
{\p \over \p u_i} L(u) = -n^{-1} \sum_{j=1}^n \sin (u_{ni}-u_{nj}).
\ee
Thus, \eqref{RAKM} is a gradient system
$$
\dot u_n =-\nabla L(u_n).
$$
and
\be\lbl{Lnegative}
\dot L =(\nabla L(u_n), \dot u_n) =- \sum_{i=1}^n \left( n^{-1} \sum_{j=1}^n 
 a_{ni} |R(u_n)|\sin(u_{ni}-\phi)\right)^2\le 0.
\ee
By the Barbashin-Krasovskii-Lassale extension of the Lyapunov's direct method 
\cite{BarKra52, Kras-Stability, LaSalle60}, we conclude that
the $\omega$-limit set of \eqref{AKM} is the set of equilibria 
$\mathcal{E}_{n,1}\cup\mathcal{E}_{n,2}$.
\end{proof}

\subsection{Stability of equilibria in $\mathcal{E}_{n,1}$}
In this subsection, to gain first insights into the asymptotic states of the 
repulsively coupled KM, we study stability of phase locked steady states of \eqref{RAKM}. 

By the definition of $\mathcal{E}_{n,1}$, for $u=(u_1,u_2,\dots, u_n)\in\mathcal{E}_{n,1}$ we have
$
u_j-u_i\in\{0, \pi\},
$
$
i,j, \in [n].
$
Thus, up to translation by a constant vector 
$$
\mathcal{E}_{n,1}=\bigcup_{m=1}^n \mathcal{E}^{(m)}_{n,1},
$$
where $\mathcal{E}^{(m)}_{n,1}$ consists of equilibria with precisely $m$
coordinates equal to $0$ and the rest to $\pi$. For example,
\be\lbl{step-m-n}
(\underbrace{0,0,\dots,0}_{m}, \underbrace{\pi,\pi,\dots,\pi}_{n-m})\in \mathcal{E}^{(m)}_{n,1}.
\ee
Suppose $u_n^{(m)}=(u_{n,1}^{(m)},u_{n,2}^{(m)},\dots, u_{n,n}^{(m)})\in  \mathcal{E}^{(m)}_{n,1}$
for some $m\in [n]$. 
Then there is an $m$-element subset $\Lambda_n^{(m)}\subset [n],$ $|\Lambda_n^{(m)}|=m$
such that
\be\lbl{def-uni}
u_{n,i}^{(m)} =\left\{ \begin{array}{cc} 0, & i\in \Lambda_n^{(m)},\\
\pi, & i\notin \Lambda_n^{(m)}.
\end{array}
\right.
\ee
Denote the matrix of linearization of \eqref{rKM} about $u_{n}^{(m)} $ by $A$.
A straightforward computation shows that 
\be\lbl{Jacob}
A=D-vv^\t,
\ee
where  $D=(d_{ij})$ is a diagonal matrix with nonzero entries
\be\lbl{Dii}
d_{ii}=\left\{ \begin{array}{cc} x_{ni}^{-\alpha} d, & i\in \Lambda_n^{(m)},\\
-x_{ni}^{-\alpha} d, & i\notin \Lambda_n^{(m)},
\end{array}
\right.
\quad d:=n^{-1}\sum_{j=1}^n (-1)^{\sigma (j)} x_{nj}^{-\alpha},
\ee
\be\lbl{rank-one}
v=\left((-1)^{\sigma (1)} x_{n1}^{-\alpha},(-1)^{\sigma (2)} x_{n2}^{-\alpha},\dots,
(-1)^{\sigma (n)} x_{nn}^{-\alpha}\right),\quad
\sigma(i)=\left\{ \begin{array}{cc}
0, & i\in \Lambda_n^{(m)},\\
1, & i\notin \Lambda_n^{(m)}.
\end{array}
\right.
\ee

\begin{lem}\lbl{lem.interlace}
Let $u\in\mathcal{E}_{n,1}^{(m)}, m\in [n].$ 

If $d>0$ then the matrix of linearization about $u^{(m)}_n$, $A$, 
has at least $m-1$ positive eigenvalues of $D$, at least  
$n-m-1$ negative eigenvalues, and at least one zero eigenvalue.

If $d<0$, $A$ has at least $n-m-1$ positive eigenvalues of $D$, at least  
$m-1$ negative eigenvalues, and at least one zero eigenvalue.
\end{lem}
\begin{cor}\lbl{cor.unstable}
All equilibria from $\mathcal{E}_{n,1}^{(m)}, 1<m\le n,$ are unstable. In particular, 
 all solutions of the form \eqref{step-m-n}
are unstable for $m>1$.
\end{cor}
\begin{proof}[Proof of Lemma~\ref{lem.interlace}.]
Suppose $d>0$.
Let $\lambda_k(D)$ and $\lambda_k(A)$ denote the eigenvalues of $D$ and $A$ arranged in the increasing
order counting multiplicity. By the interlacing inequalities (cf.~\cite[Theorem~4.3.4]{HornJohnson}),
we have 
$$
\lambda_{n-m-1}(A)\le \lambda_{n-m} (D)<0,
$$
 and
$$
0<\lambda_{n-m+1}(D)\le \lambda_{n-m+2}(A).
$$
Finally, there is at least one zero eigenvalue of $A$, because its rows sum to $0$.

The case $d<0$ is analyzed similarly.
\end{proof}

\subsection{The Ott-Antonsen Anzats}
Unlike equilibria in $\mathcal{E}_{n,1}$ considered in the previous subsection, 
equilibria in $\mathcal{E}_{n,2}$ are harder to identify explicitly. To study properties
of the equilibria in $\mathcal{E}_{n,2}$, we will invoke the mean field
equation: 
\be\lbl{rMF}
{\p \over \p t} \rho(t,u,x) +{\p \over \p u} \left\{ V(t,u,x) \rho(t,u,x)\right\}=0,
\ee
where 
\be\lbl{rV}
V(t,u,x)= \int_I \int_\SS W(x,y) \sin (u-v) \rho(t,v,y) dv dy.
\ee

To identify a class of stable steady states of \eqref{rMF}, we employ
the Ott-Antonsen Anzats \cite{OttAnt08}, i.e., we look for solutions of \eqref{rMF} in the 
following form:
\be\lbl{OA}
\rho(t,u,x)={1\over 2\pi} \left( 1+ \sum_{k=1}^\infty \left(
\overline{z(t,u)}^k e^{\iu ku} +
 z(t,u)^k e^{-\iu ku} \right) \right).
\ee
If $\sup_{x\in I, t\in [0,T]} |z(t,x)|<1$ the series on the right hand side
of \eqref{OA} is absolutely convergent. Plugging \eqref{OA} into 
\eqref{rMF}, after straightforward albeit tedious manipulations, one 
verifies that \eqref{OA} solves \eqref{rMF}, provided $z(t,x)$ satisfies the 
following equation
\be\lbl{zdot}
{\p \over \p t} z(t,x)={1\over 2 x^\alpha} \left(z(t,x)^2-1\right) \mathcal{R}[z(t,\cdot)],
\ee
where
\be\lbl{corder}
\mathcal{R}[v]=\int_I y^{-\alpha} v(y)dy.
\ee
From \eqref{OA}, it follows that 
\be\lbl{inv-z}
z(t,x)= \int_0^{2\pi} \rho(t,u,x) e^{\iu u} du.
\ee
Using \eqref{inv-z}, we can express $\mathcal{R}[z]$ in
terms  of the density $\rho$:
\be\lbl{corder-1}
\mathcal{R}[z]= \int_I\int_\SS y^{-\alpha} \rho(t,u,y) e^{\iu u} dudy,
\ee
i.e., $\mathcal{R}$ is the continuous counterpart
 of the order parameter \eqref{order}.

In the remainder of this section, we restrict to real solutions of \eqref{zdot}.
It is instructive to review the interpretation of $z(t,x)$ (cf.~\cite{Ome13}). To this end, 
note that \eqref{OA}
implies 
\be\lbl{Poisson}
\rho(t,u,x)={ 1-|z(t,x)|^2\over 
2\pi\left( 1-2 |z(t,x)| \cos (u-\operatorname{Arg}~z(t,x) )
+ |z(t,x)|^2\right)}.
\ee
In particular, $z\equiv 0$ corresponds to the uniform density, while
values of $z$ close to $\pm 1$ indicate that the density is concentrated around
$0$ and $\pi$ respectively.

\subsection{The nonlocal equation}

Let $\mathcal{M}(0,1)$ be a space of measurable functions 
$z: (0,1]\to [-1,1]$.
In analogy to the discrete model \eqref{rKM}, we divide the equilibria of \eqref{zdot} into two classes:
$$
\tilde{\mathcal{E}}_1=\left\{ z\in \mathcal{M}(0,1):\; (|z(x)|=1,\; x\in (0,1])\& 
(\mathcal{R}[z]\neq 0) \right\},\quad
\tilde{\mathcal{E}}_2=\left\{ z\in \mathcal{M}(0,1):\; \mathcal{R}[z]= 0\right\}.
$$
\begin{thm}\lbl{thm.nonlocal}
Let $z_0\in \mathcal{M}(0,1)$ and $z_0\notin \tilde{\mathcal{E}}_1$.
Then the $\omega$--limit set of the trajectory of \eqref{zdot} starting at $z_0$,
$\mathbf{\omega}(z_0)\in \tilde{\mathcal{E}}_2$.
\end{thm}
\begin{proof}
Suppose $\mathcal{R}(0):=\mathcal{R}[z(0,\cdot)] >0$. Changing the time variable to
$t=\phi(\tau)$ subject to
\be\lbl{ch-time}
\phi^\prime(\tau)= {1\over \int_I y^{-\alpha} \tilde z(\tau,y) dy},\quad
\tilde z(\tau,y):=z(\phi(\tau),x),\quad \phi(0)=0,
\ee
we reduce \eqref{zdot} to
\be\lbl{tilde-z}
{\p\over \p\tau} \tilde z= {1\over 2x^\alpha} \left(\tilde z^2 -1\right).
\ee
The last equation is integrated explicitly
\be\lbl{separable}
\tilde z(\tau, x)={1-C(x)\exp\{{\tau\over  x^\alpha}\}\over 
1+C(x)\exp\{{\tau\over x^\alpha}\}},\quad
C(x)={1-z_0(x)\over 1+z_0(x)}.
\ee
Clearly, $\tilde z(\tau,x)\searrow -1$ for $x\in I$ as $\tau\to\infty$. 
Thus, there is $0<\tau^\ast<\infty$ such that
$\tilde{\mathcal{R}}(\tau^\ast):=\tilde{\mathcal{R}}[\tilde z(\tau^\ast,\cdot)]=0$ and 
$\tilde{\mathcal{R}}(\tau)>0$ for $\tau\in [0,\tau^\ast)$.

The change of time \eqref{ch-time} is well defined for $\tau\in [0,\tau^\ast).$ In terms of the original
time, we have the description of the system's dynamics on the time interval $[0,t^\ast)$,
with
\be\lbl{t-star}
t^\ast=\lim_{\tau\to\tau^\ast-0} \int_0^\tau (\tilde{\mathcal{R}} (s))^{-1} ds.
\ee

Denote $
\mathcal{R}(t)=\mathcal{R}[z(t,\cdot)].
$
If $t^\ast<\infty$ then $\mathcal{R}(t)=\tilde{\mathcal{R}}(\tau^\ast)=0$ for $t\ge t^\ast$.
Otherwise,
multiplying both sides of \eqref{zdot} by $x^{-\alpha}$ and integrating over
$I$, we have
\be\lbl{Rdot}
\begin{split}
\mathcal{R}^\prime& =2^{-1} \int_I x^{-2\alpha} (z^2(x,t)-1) dx \,\mathcal{R}\\
&\le 2^{-1} \int_I x^{-2\alpha} (z(x,t)-1) dx \, \mathcal{R}\\
&\le 2^{-1} \int_I x^{-\alpha} (z(x,t)-1) dx \,\mathcal{R}\\
& \le  \left(\mathcal{R} - (1-\alpha)^{-1}\right) \mathcal{R} .
\end{split}
\ee
By the comparison principle (cf.~\cite[Theorem~I.4.1]{Hartman-ODEs}),
from \eqref{Rdot} we conclude that $\mathcal{R}(t)\searrow 0$, as $t\to\infty$.

The case $\mathcal{R}[z(0,\cdot)]< 0$ is analyzed similarly.
\end{proof}

\subsection{Attractors of the repulsively coupled model}
Theorem~\ref{thm.nonlocal} shows that solutions of the mean field equation 
of the form \eqref{OA} approach an equilibrium
from the set $\{\mathcal{R}=0\}$. To illustrate possible patterns generated in this
scenario, we consider the 
IVP for \eqref{zdot} with the following initial conditions:
\be\lbl{ic}
\begin{split}
z_{\delta,x_0}^{(step)} (x) &=\left\{\begin{array}{cc}
-1+\delta, & x\in I^-:=(0,x_0),\\
1-\delta, & x\in I^+:=[x_0, 1),
\end{array} 
\right.
\end{split}
\ee

For $0<\delta\ll 1$, $z_{\delta,x_0}^{(step)}$ is close to the equilibrium $z_{0,x_0}^{(step)}\in
\tilde{\mathcal{E}}_1$,
corresponding
to a phase locked solution, localized around
$\pi$ for $x\in I^-$ and around $0$ for $x\in I^+$.

Consider the IVP for \eqref{zdot} with initial condition
$z_{\delta,x_0}^{(step)}$ for $0<\delta< 1$. To this end, note that for 
$x^\ast=2^{-1\over 1-\alpha}\in (0,1),$ 
$$
\int_0^{x^\ast} y^{-\alpha} dy = - \int_{x^\ast}^1 y^{-\alpha}dy.
$$
\begin{figure}
\begin{center}
{\bf a}\hspace{0.1 cm}\includegraphics[height=1.8in,width=2.0in]{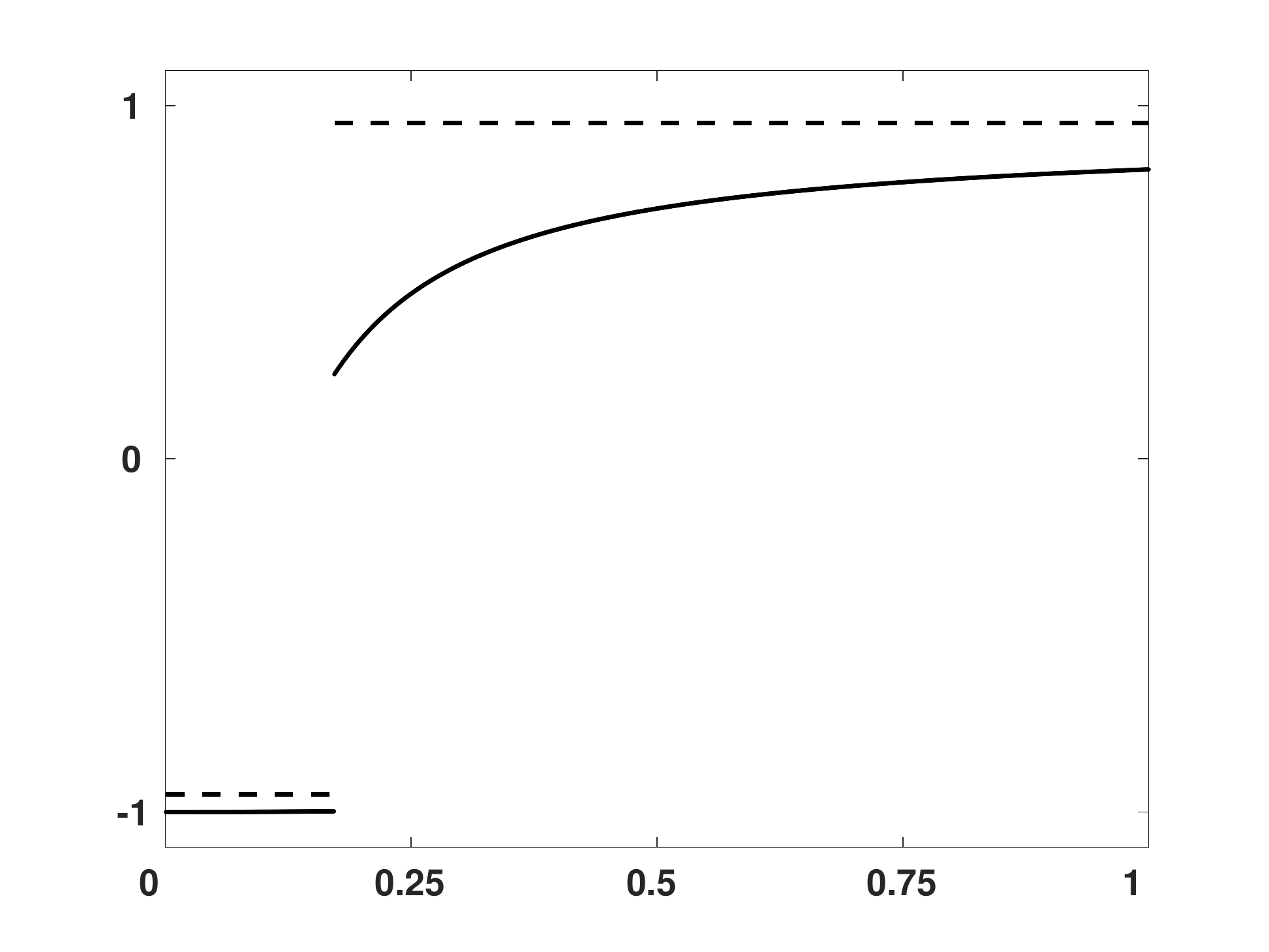}
{\bf b}\hspace{0.1 cm}\includegraphics[height=1.8in,width=2.0in]{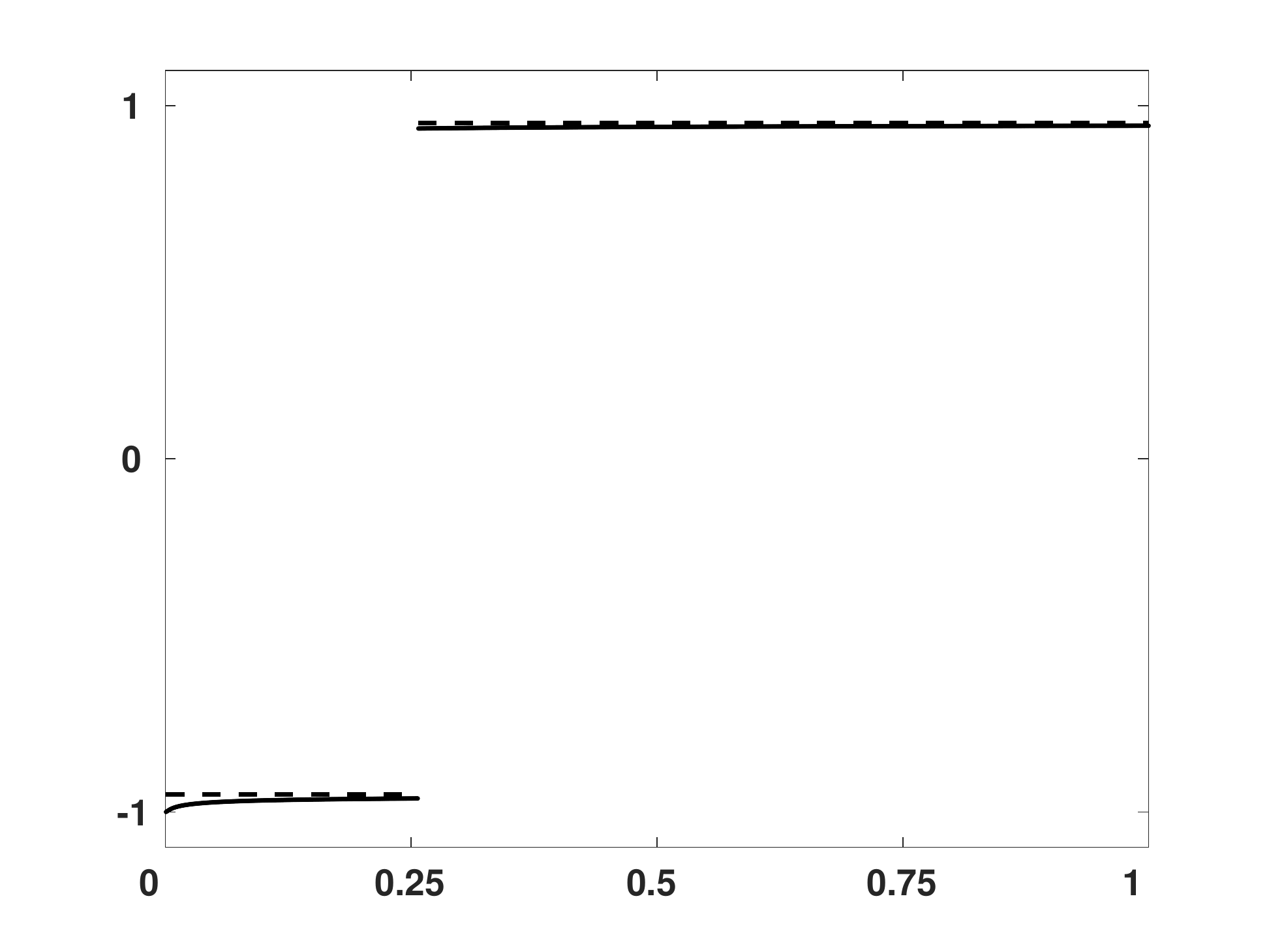}
{\bf c}\hspace{0.1 cm}\includegraphics[height=1.8in,width=2.0in]{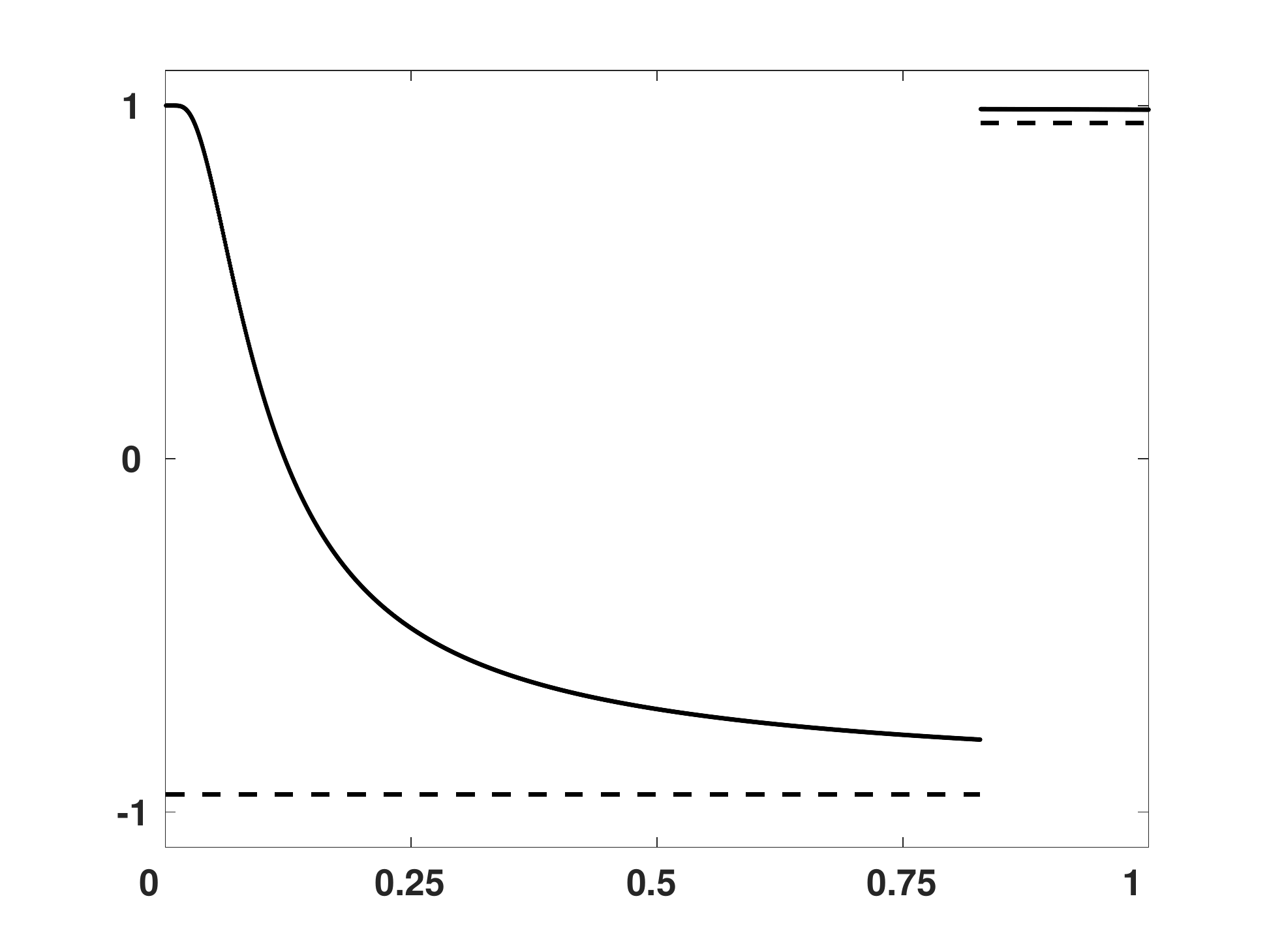}\\
{\bf d}\hspace{0.1 cm}\includegraphics[height=1.8in,width=2.0in]{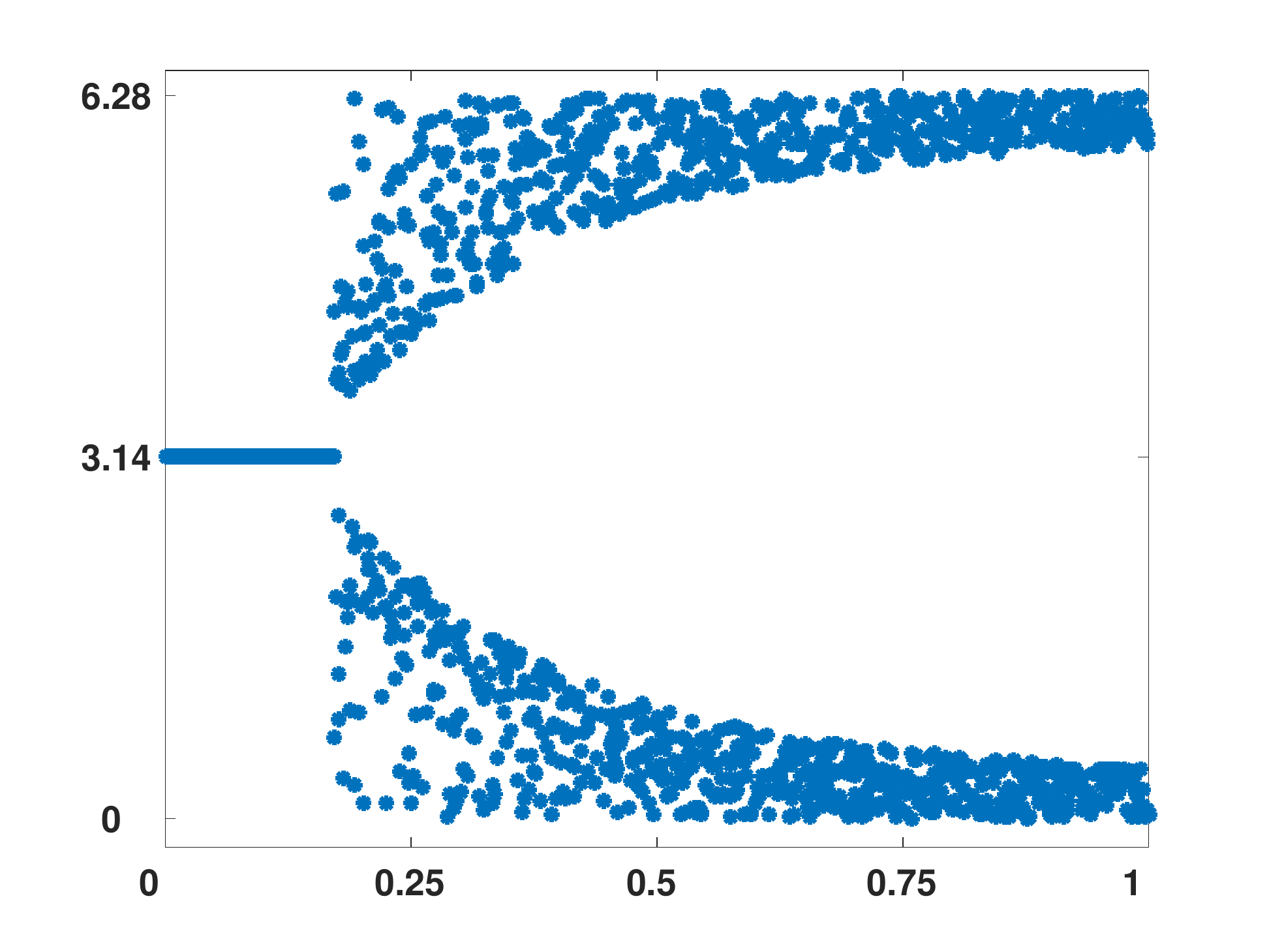}
{\bf e}\hspace{0.1 cm}\includegraphics[height=1.8in,width=2.0in]{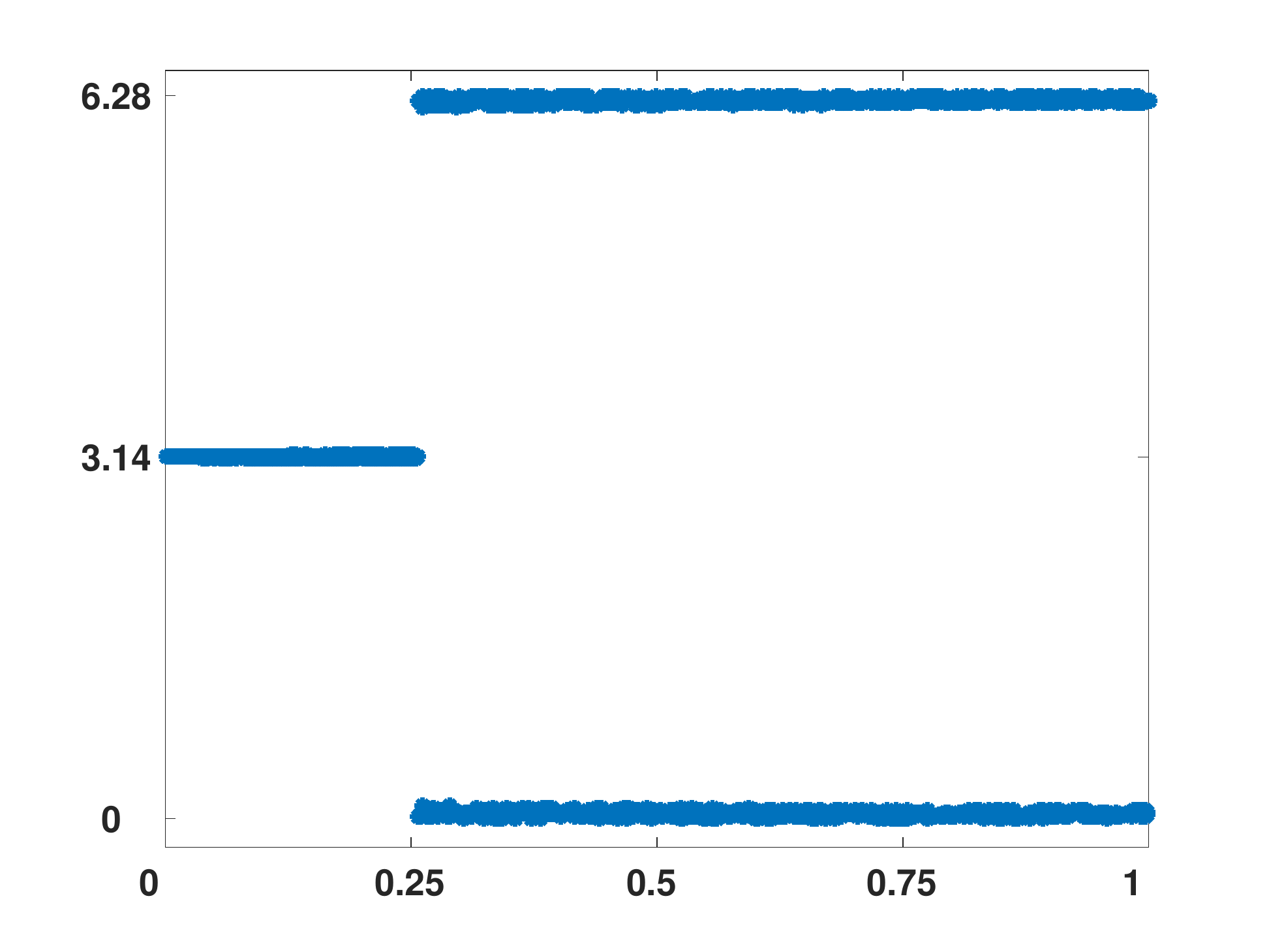}
{\bf f}\hspace{0.1 cm}\includegraphics[height=1.8in,width=2.0in]{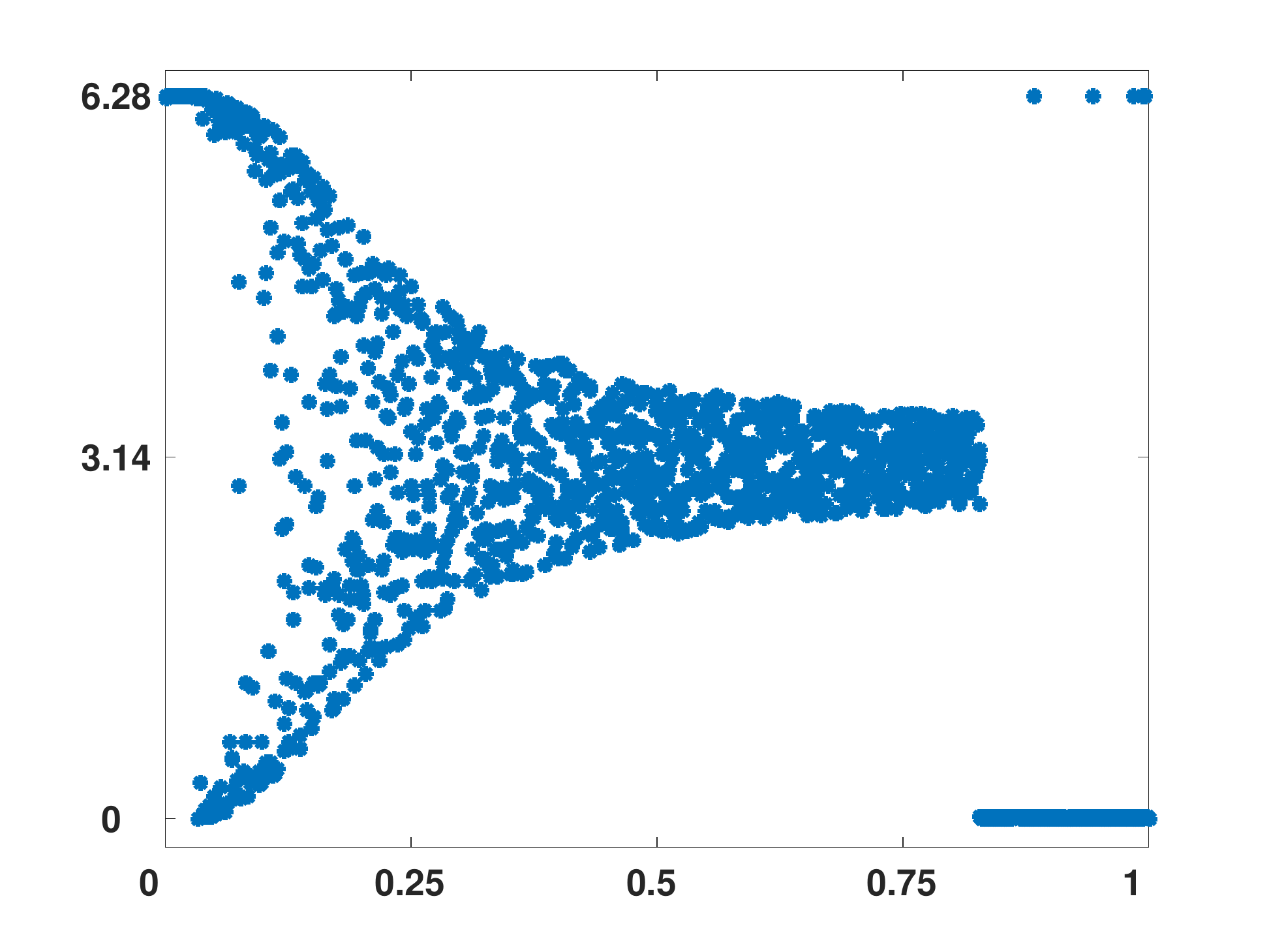}
\end{center}
\caption{The initial conditions (dashed line) and asymptotic states (solid line) 
for \eqref{zdot} (\textbf{a}-\textbf{c}) and the corresponding steady states of \eqref{rAKM}
(\textbf{d}-\textbf{f}). The patterns shown in (\textbf{d}) and (\textbf{f}) are the examples
of chimera states.
In (\textbf{d}) the oscillators in the left region ($I^-$) are localized around $\pi$ and
are spread-out around $0$ in the right region $I^+$. Similarly, the 
pattern shown in (\textbf{f}) features localized distribution around $0$ and
the spread-out one around $\pi$.
}
\lbl{f.1}
\end{figure}

Suppose first that $0<x_0<x^\ast$. Then $\mathcal{R}(0)>0$. By Theorem~\ref{thm.nonlocal},
$\mathcal{R}(t)\searrow 0$.
Furthermore, $|z(t,x)|\le 1$ and, thus, $z(t,x)$ is monotonically decreasing in time for 
every $x\in (0,1]$. In particular,
\be\lbl{left-I}
-1\le z(t,x)\le -1+\delta, \quad x\in I^-=(0,x_0),\; t\ge 0.
\ee
This means that in $I^-$ the oscillators remain localized around $\pi$ (in the moving frame of coordinates)
(cf.~\eqref{Poisson}),
provided $0<\delta\ll 1$ (Fig.~\ref{f.1}d). On the other hand, in $I^+$, $z(t,\cdot)$ is monotonically decreasing 
to its asymptotic state $z_\infty$ at which $\mathcal{R}(z_\infty)=0$  (see Fig.~\ref{f.1}\textbf{a}). 
In $I^+$, there must be
an interval over which $z$ is positive and strictly less than $1-\delta$ for all times. Denote such interval 
$\tilde I^+\subset I^+$. Thus, over $\tilde I^+$, $z(t,x)$ is bounded away from $\pm 1$ by a distance
greater than $\delta$ uniformly 
in time. Thus, the oscillators over $\tilde I^+$ exhibit a greater degree of 
incoherence.
The asymptotic state $z_\infty$ contains both the region of 
coherent dynamics ($I^-$) and that of incoherent ($I^+$) (Fig.~\ref{f.1}\textbf{a}).
Thus, $z_\infty$ corresponds to a chimera state. This is clearly seen in numerics  
(see Fig.~\ref{f.1}\textbf{d}). 
In the next section, for the modified model we will present tight estimates characterizing the asymptotic
state $z_\infty$. 
  
Next, we comment on the transformation of the asymptotic state $z_\infty$ as $x_0$ 
is increasing past $x^\ast$. The case of $x>x^\ast$ presents a symmetric scenario. 
In this case, $\mathcal{R}(0)<0$
and both $\mathcal{R}(t)$ and $z(t,\cdot)$ are monotonically increasing (see Fig.~\ref{f.1}\textbf{c}).  
In particular, $1-\delta\le z(t,x)\le 1$
in $I^+$ for $t\ge 0$, and the oscillators are localized around $0$ in $I^+$, while
exhibiting incoherent behavior in $I^-$ (see Fig.~\ref{f.1}\textbf{c},\textbf{f}). 
When $x_0$ is close to $x^\ast$, $\mathcal{R}(0)$ is close to zero, and $|\mathcal{R}(t)|$ 
remains small for all times. This means that the initial pattern does not change much
in the process of evolution, and $z_\infty$ remains close to the step function
$$
z_{0, x_0}^{(step)}(x) :=\pm 1, \; x\in I^\pm.
$$ 
The equilibrium $z_{0, x_0}^{(step)} \in \tilde{\mathcal{E}}_1$ is unstable but lies close a stable equilibrium 
$z_{step, x^\ast}\in\tilde{\mathcal{E}}_2$ (see Fig.~\ref{f.1}\textbf{b},\textbf{e}).

\subsection{Chimera states in the modified KM}
We now turn to a modification of the KM on power law graphs, for which we derive tight estimates
for the chimera states. If instead of scaling the coupling term by $n\rho_n$, as in \eqref{KM},
we scale it by the expected degree of node $i$:
$$
d_{ni}=\E_\omega \operatorname{deg}_{\Gamma_n}(i)=\sum_{j=1}^n \bar W_{nij},
$$
the repulsively coupled KM and the corresponding averaged equation 
take the following form:
\be\lbl{mKM}
\dot u_{ni} ={1\over d_{ni}} \sum_{j=1}^n \xi_{nij}(\omega) \sin (u_{ni}-u_{nj}),\quad i\in [n],
\ee
and 
\be\lbl{amKM}
\dot v_{ni} ={1\over n} \sum_{j=1}^n  x_{nj}^{-\alpha} \sin (v_{ni}-v_{nj}),\quad i\in [n],
\ee
respectively.
The mean field equation then becomes (cf.~\cite[Example 2.5]{KVMed17})
\be\lbl{mMF}
{\p\over \p t}\rho (t,u,x) + 
{\p\over \p u}\left\{\rho (t,u,x) \int_I\int_S (1-\alpha) y^{-\alpha} \sin (u-v) \rho (t,v,y) dvdy\right\}.
\ee
Applying the Ott-Antonsen Anzats to the model at hand, we arrive at
\be\lbl{mzdot}
\dot z={1-\alpha\over 2} \left(z^2-1\right) \mathcal{R}[z].
\ee
For \eqref{mzdot} subject to the initial condition \eqref{ic}, below we present tight bounds for the large time asymptotic state
$z_\infty$.
\begin{figure}
\begin{center}
{\bf a}\includegraphics[height=1.8in,width=2.0in]{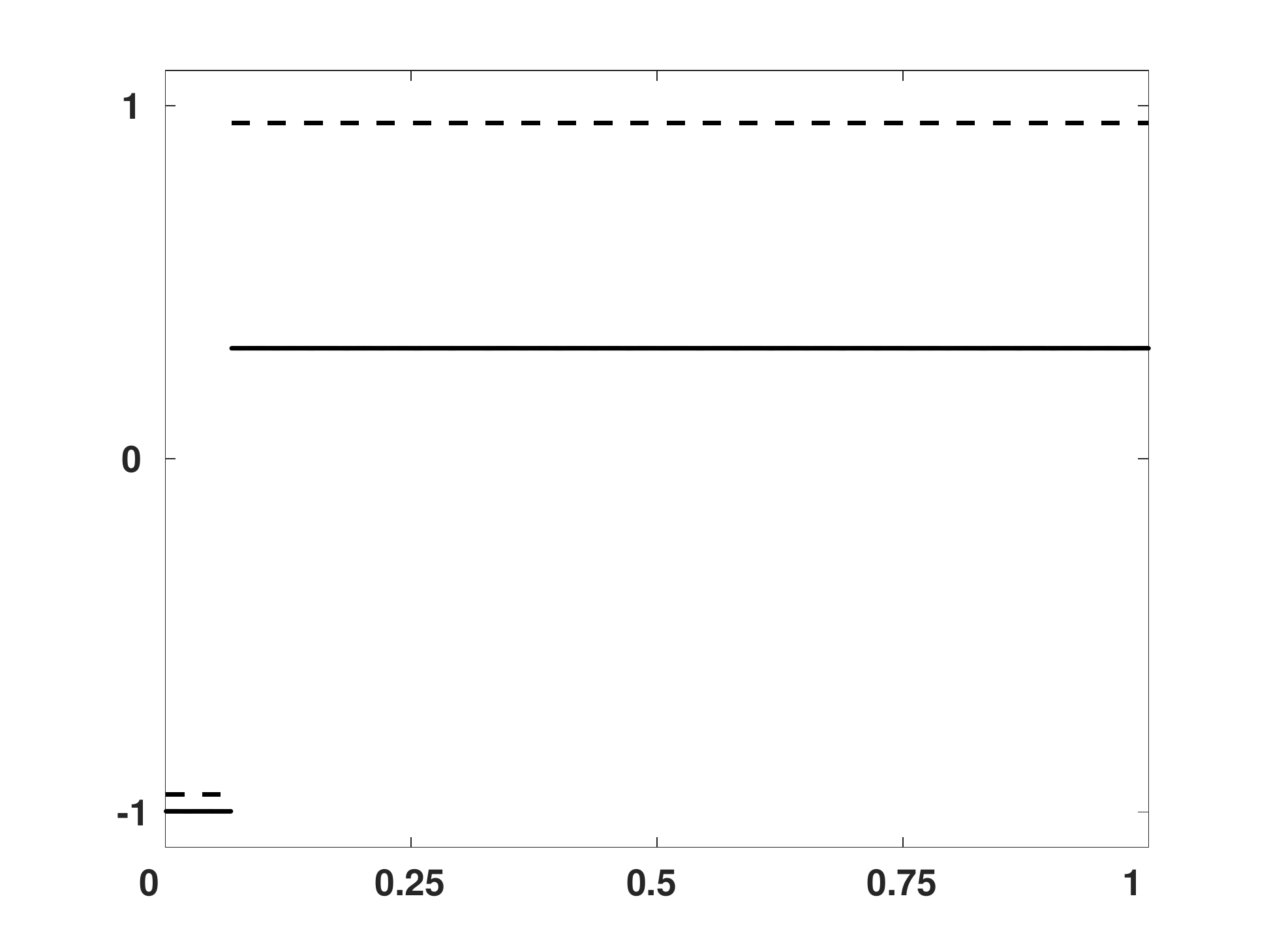}
{\bf b}\includegraphics[height=1.8in,width=2.0in]{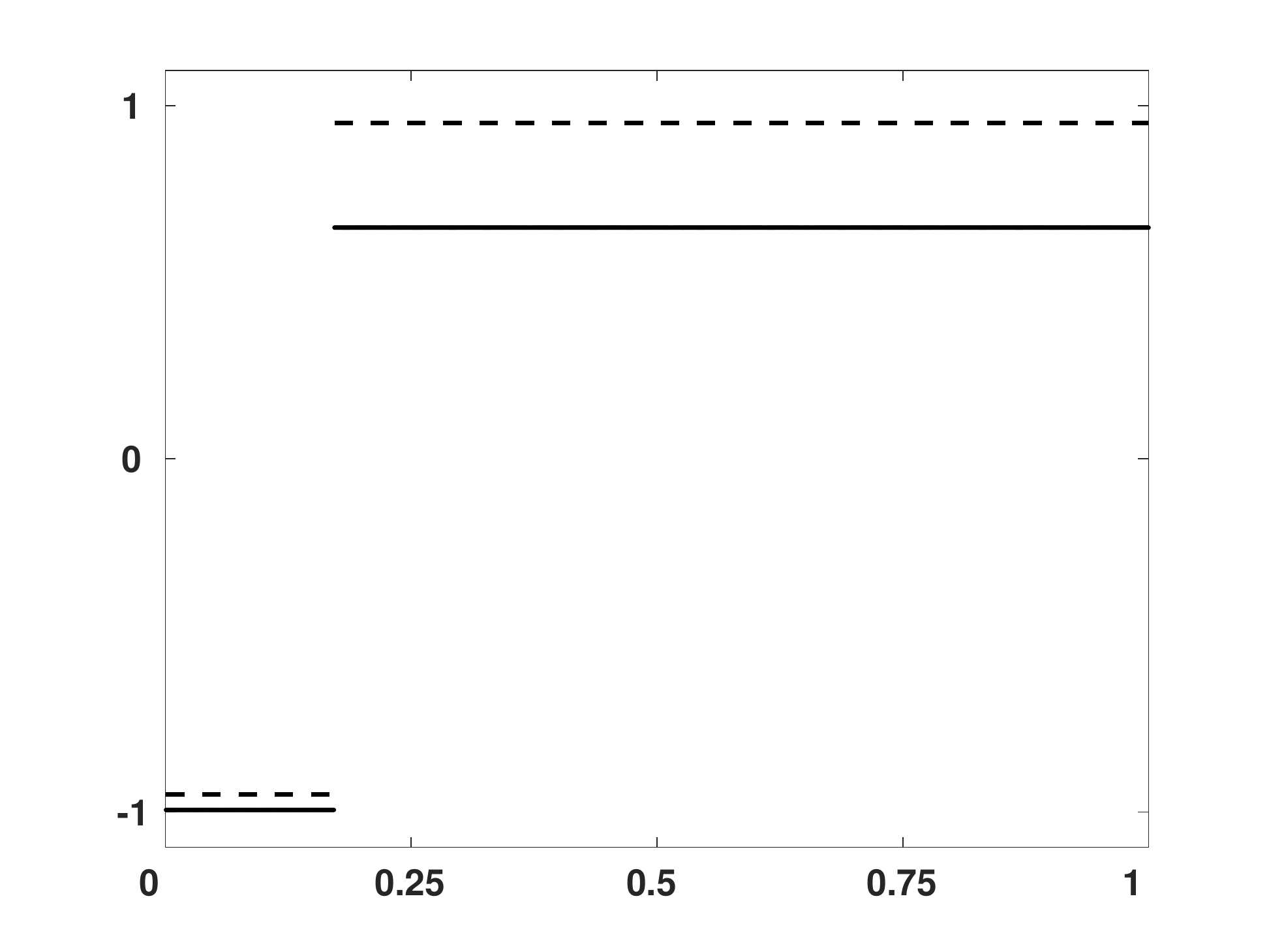}
{\bf c}\includegraphics[height=1.8in,width=2.0in]{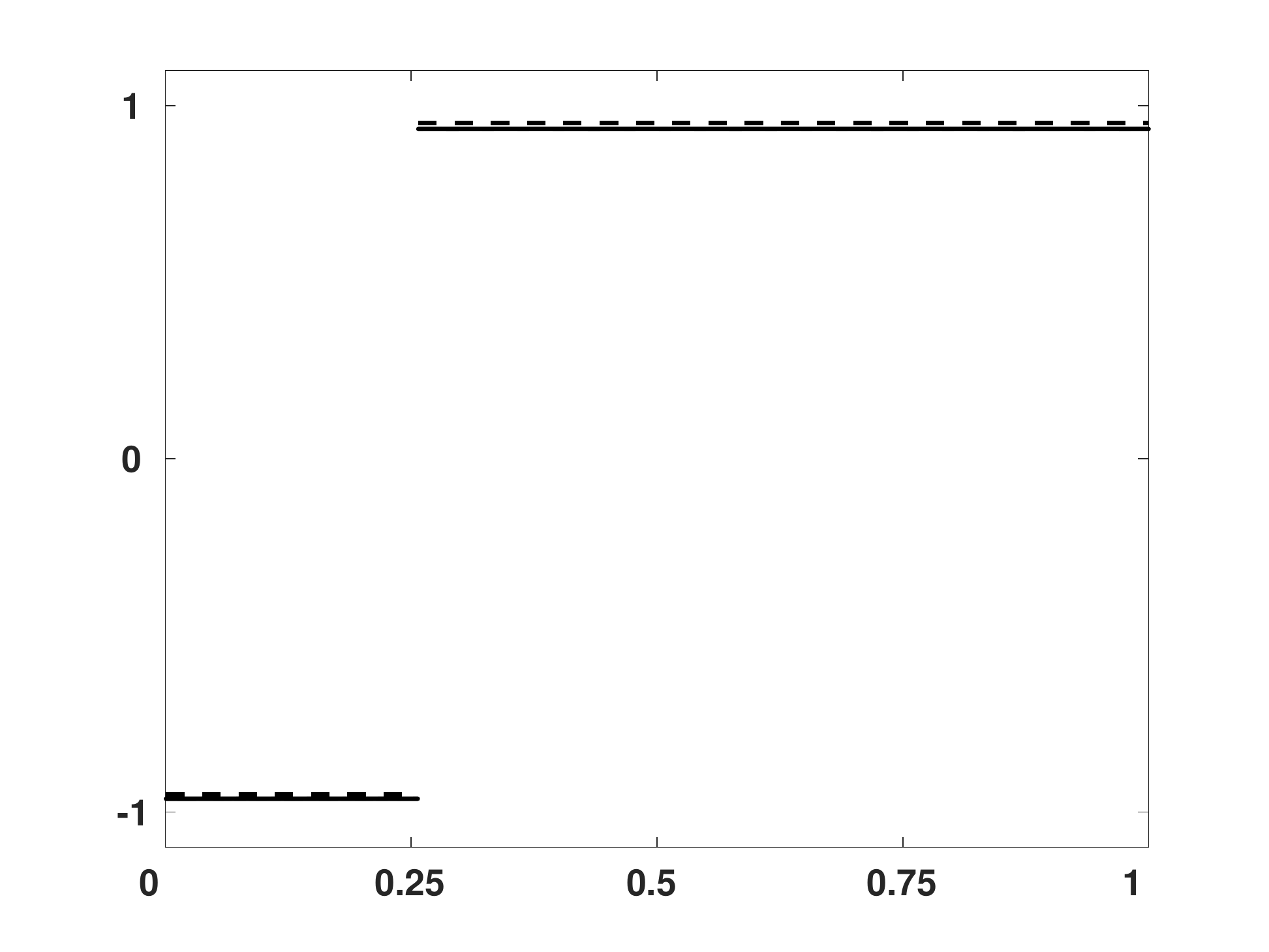}\\
{\bf d}\includegraphics[height=1.8in,width=2.0in]{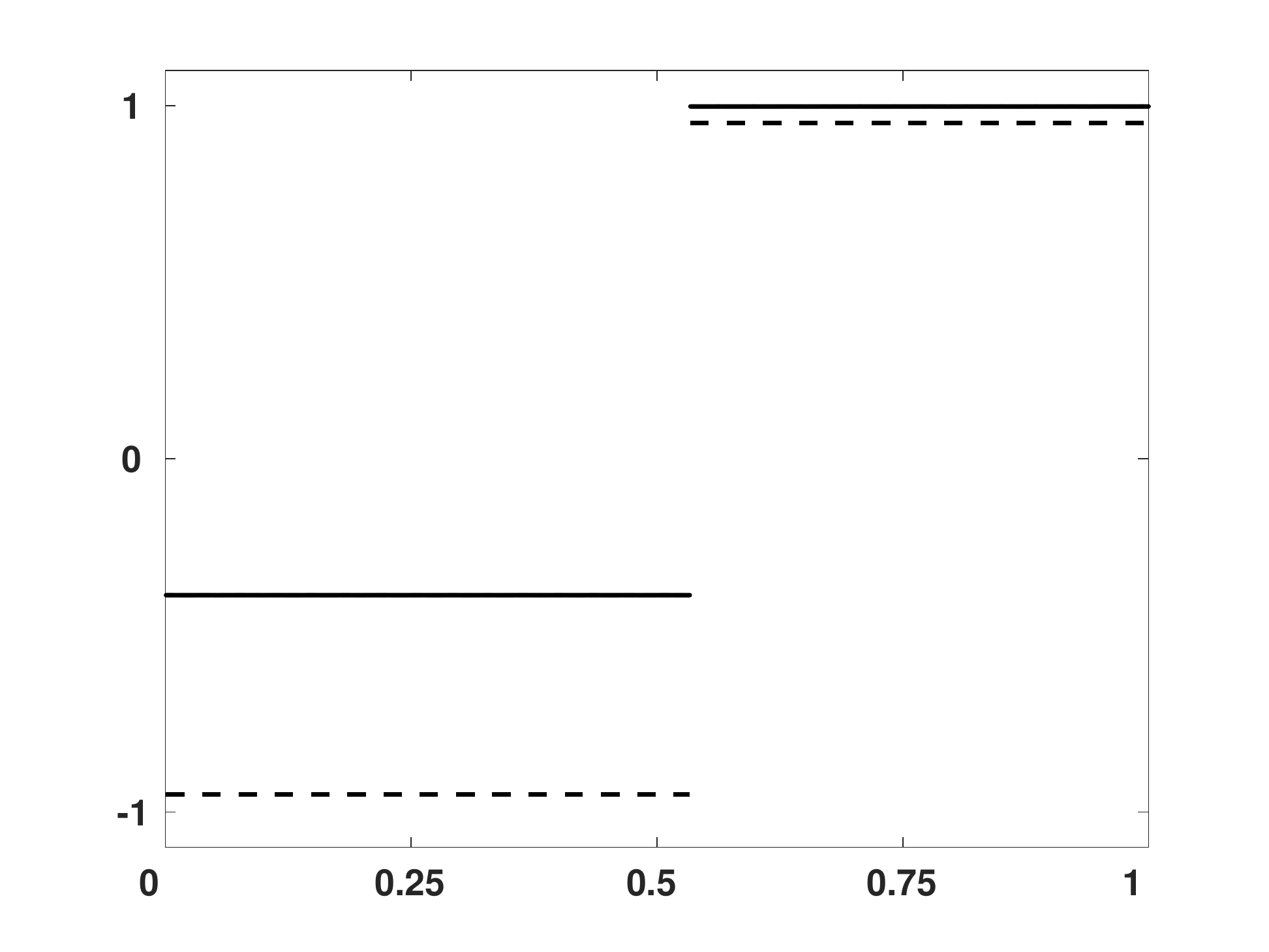}
{\bf e}\includegraphics[height=1.8in,width=2.0in]{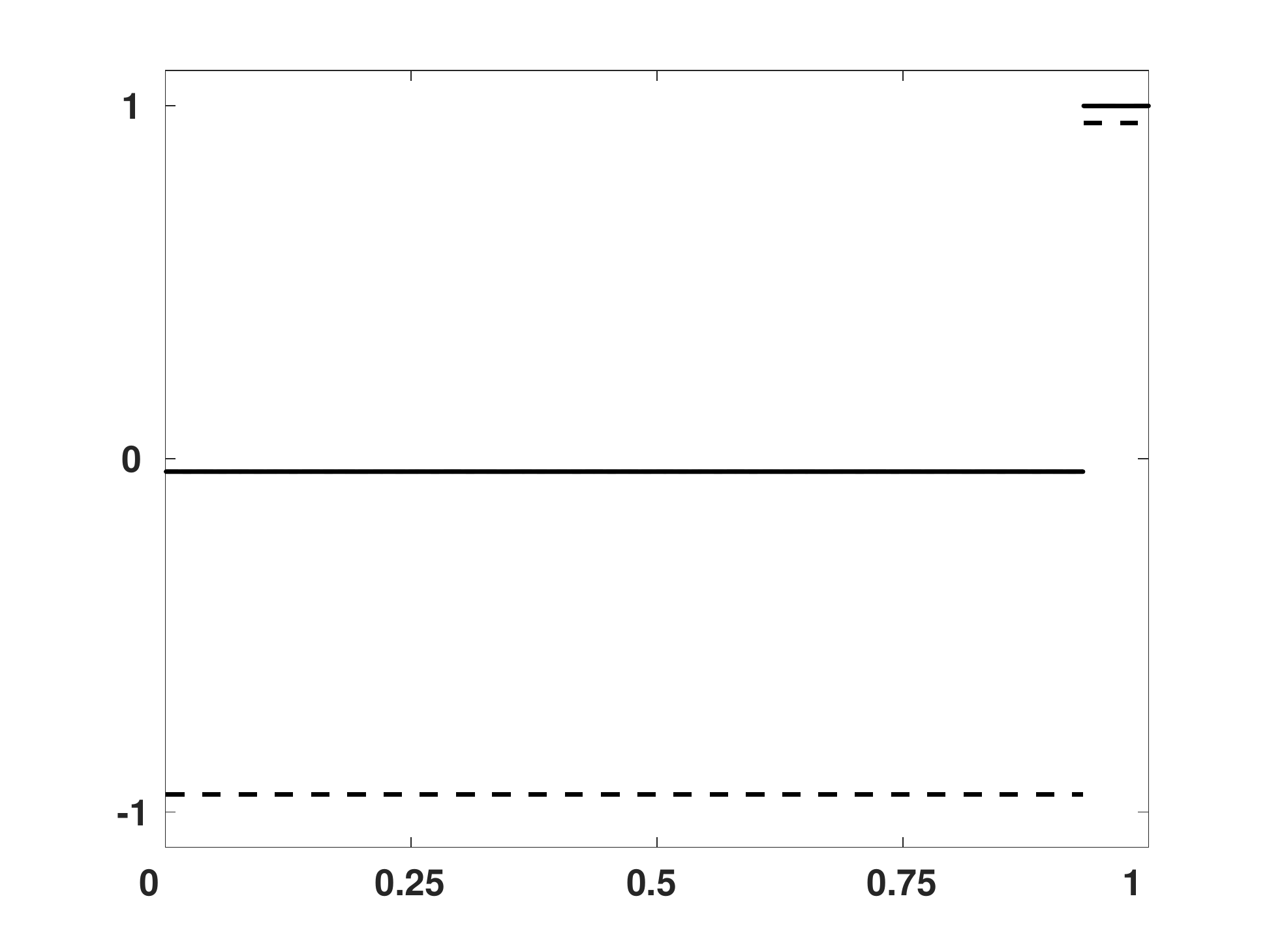}
\end{center}
\caption{The initial conditions (dashed line) and asymptotic states (solid line) 
for \eqref{mzdot}.
}
\lbl{f.2}
\end{figure}

Suppose $\mathcal{R}[z(0,\cdot)]>0$ and note that \eqref{mzdot} and \eqref{ic} imply 
$$
|z(x,t)|\le 1\quad \mbox{and}\quad  \mathcal{R}[z(t,\cdot)]\ge 0,
$$
for any $x\in I$ and $t\ge 0$.
Furthermore, $z(t,x)$ is monotonically decreasing.

On the other hand, from \eqref{mzdot} we have 
\be\lbl{bound-dotz}
(1-\alpha)(z-1) \mathcal{R}[z(t,\cdot)]\le {\p\over \p t} z\le {1-\alpha \over 2} (z-1) \mathcal{R}[z(t,\cdot)].
\ee
Multiplying all sides sides of the double inequality \eqref{bound-dotz} by $x^{-\alpha}$ and integrating over
$I$, we have
\be\lbl{bound-R}
\left((1-\alpha) \mathcal{R}[z(t,\cdot)] -1\right)\le {\p\over \p t} \mathcal{R}[z(t,\cdot)]\le 
2^{-1}\left( (1-\alpha)\mathcal{R}[z(t,\cdot)] -1\right).
\ee
\begin{figure}
\begin{center}
{\bf a}\includegraphics[height=1.8in,width=2.0in]{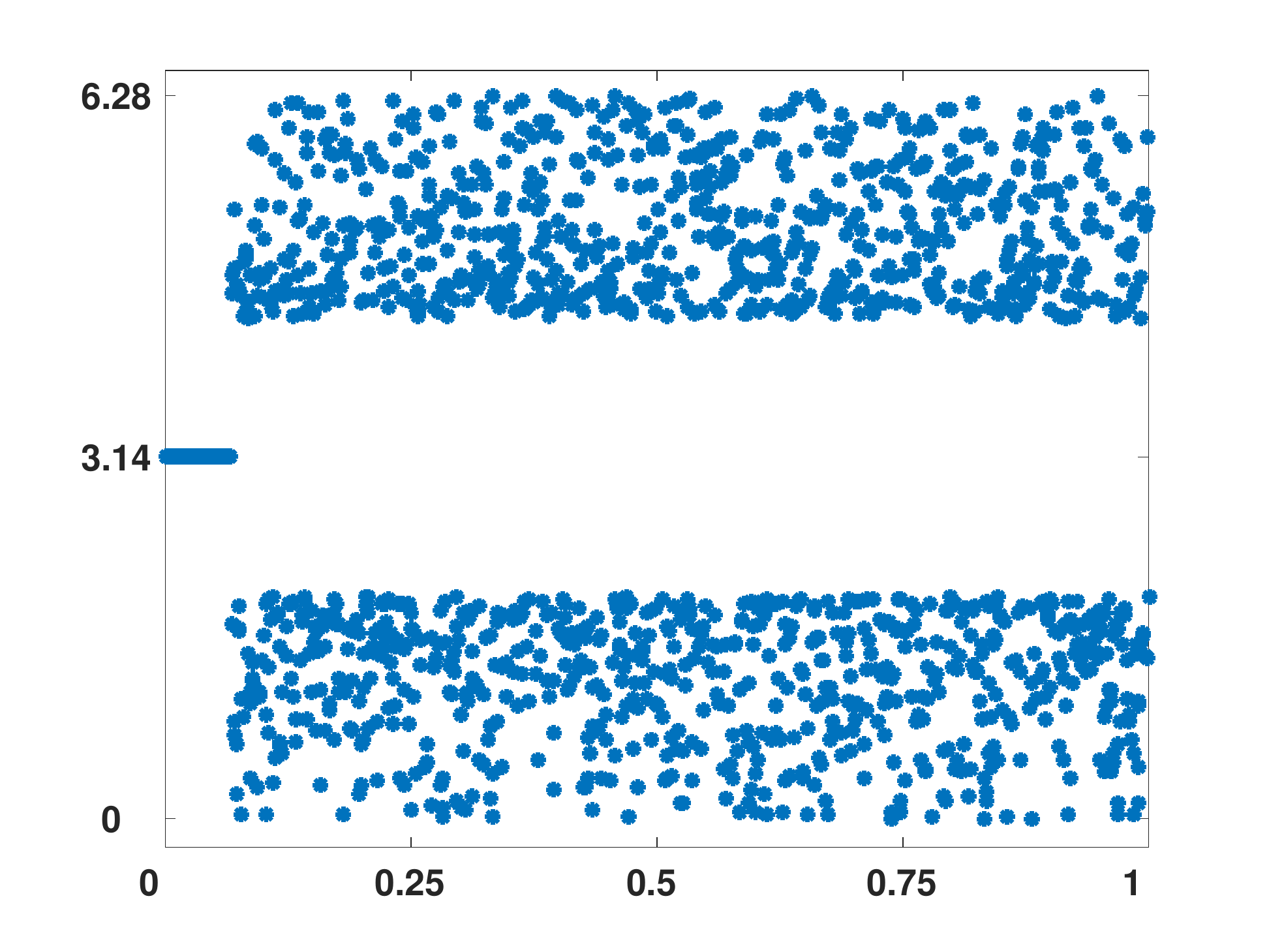}
{\bf b}\includegraphics[height=1.8in,width=2.0in]{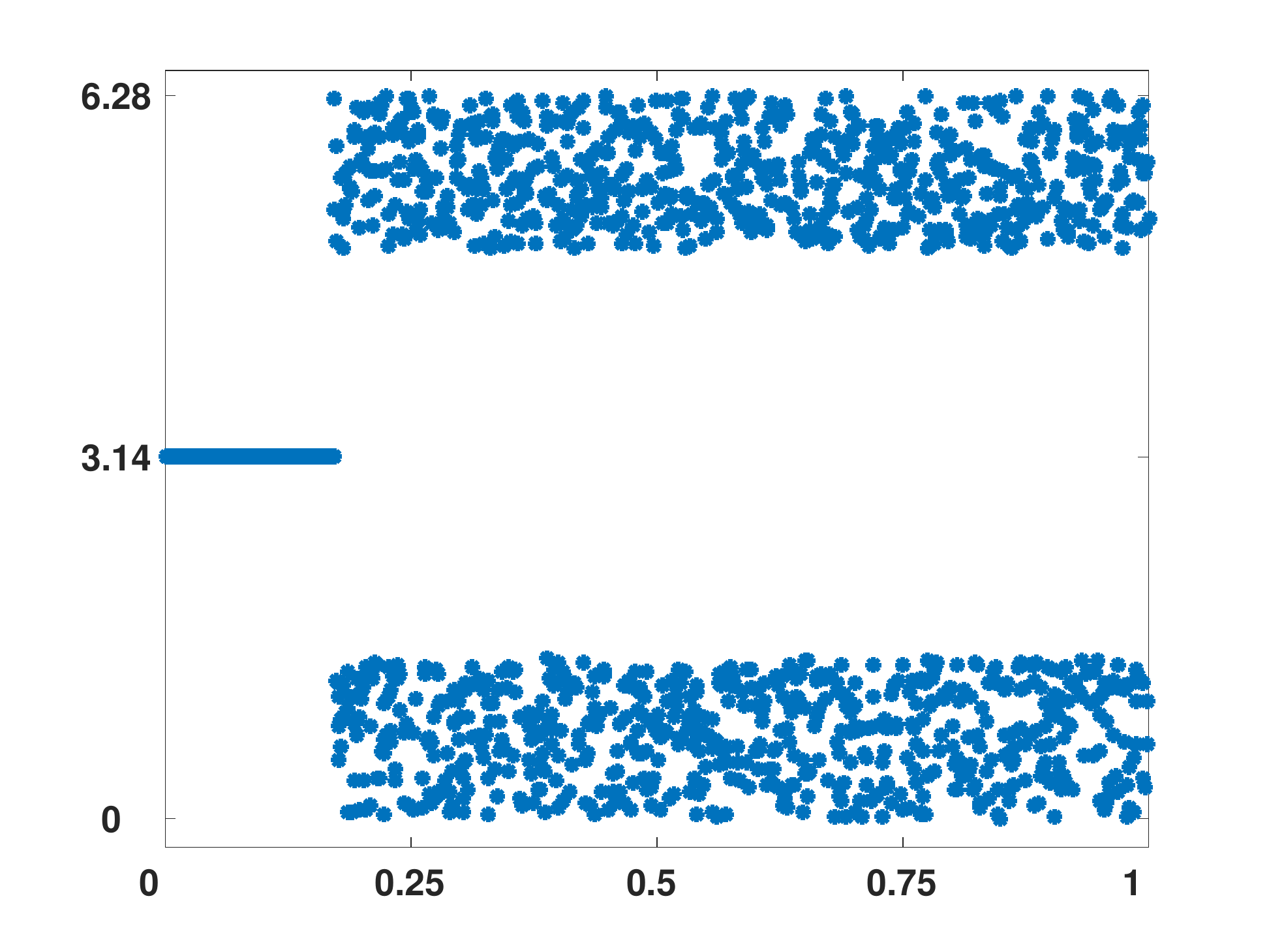}
{\bf c}\includegraphics[height=1.8in,width=2.0in]{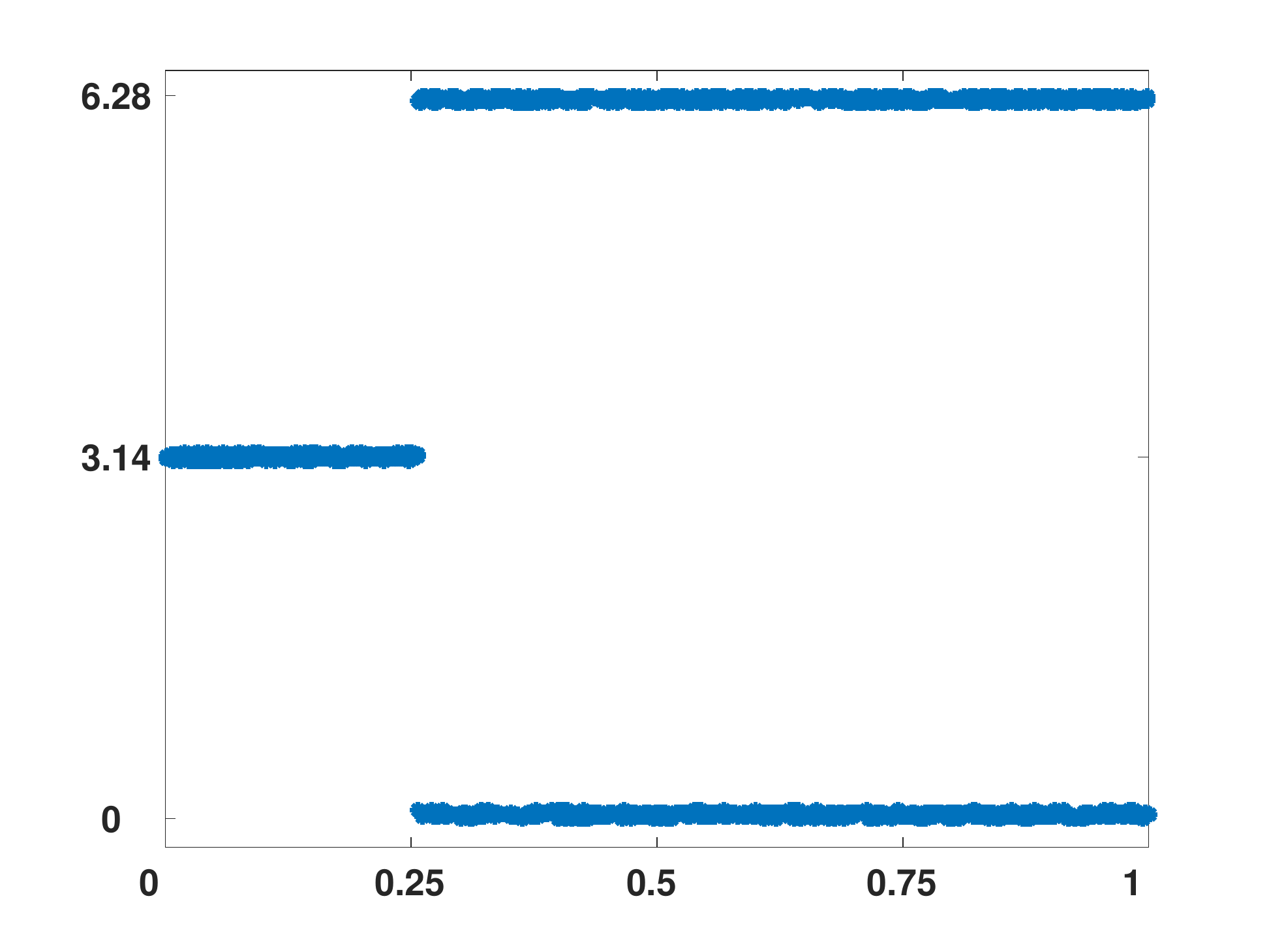}\\
{\bf d}\includegraphics[height=1.8in,width=2.0in]{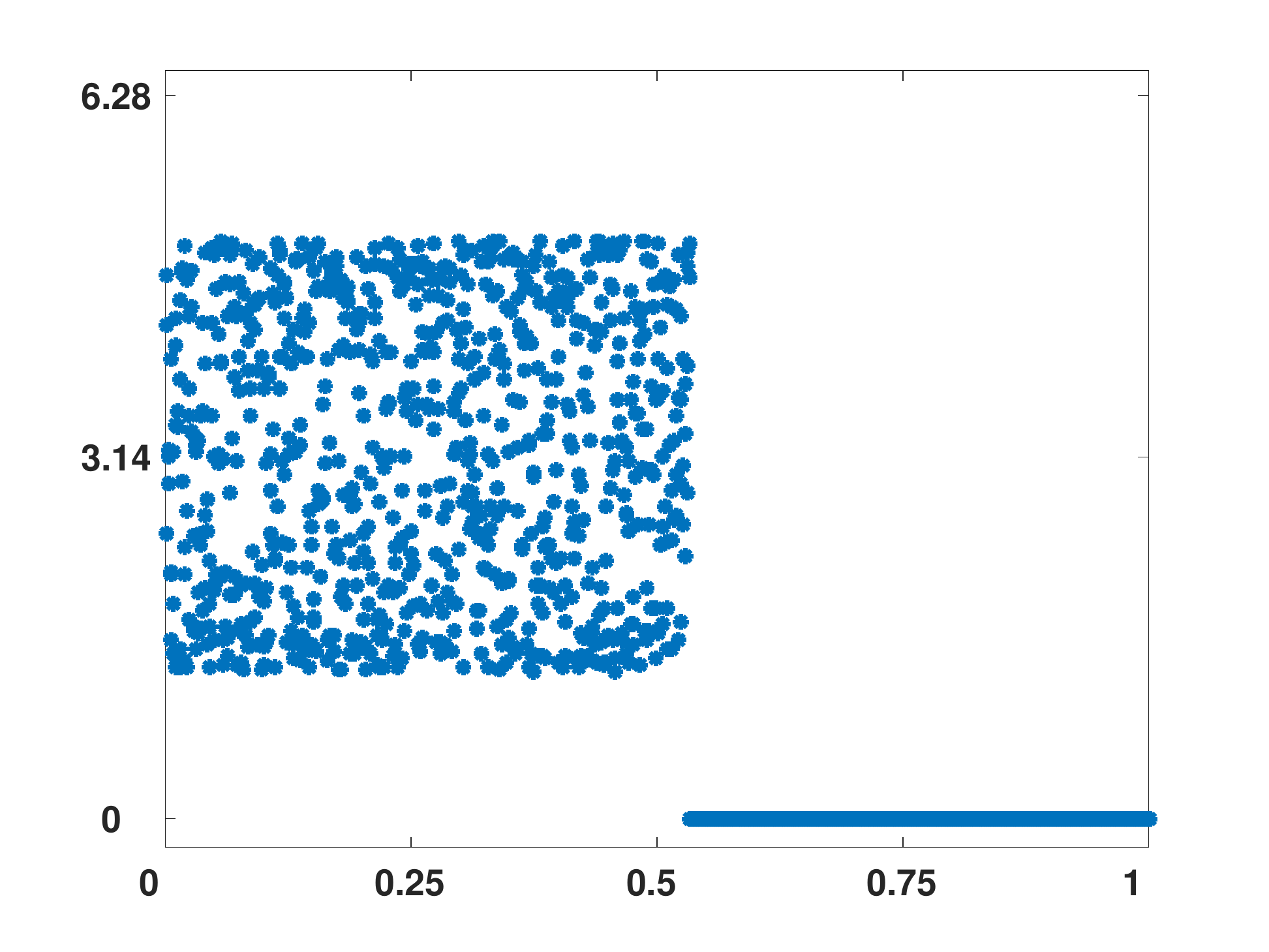}
{\bf e}\includegraphics[height=1.8in,width=2.0in]{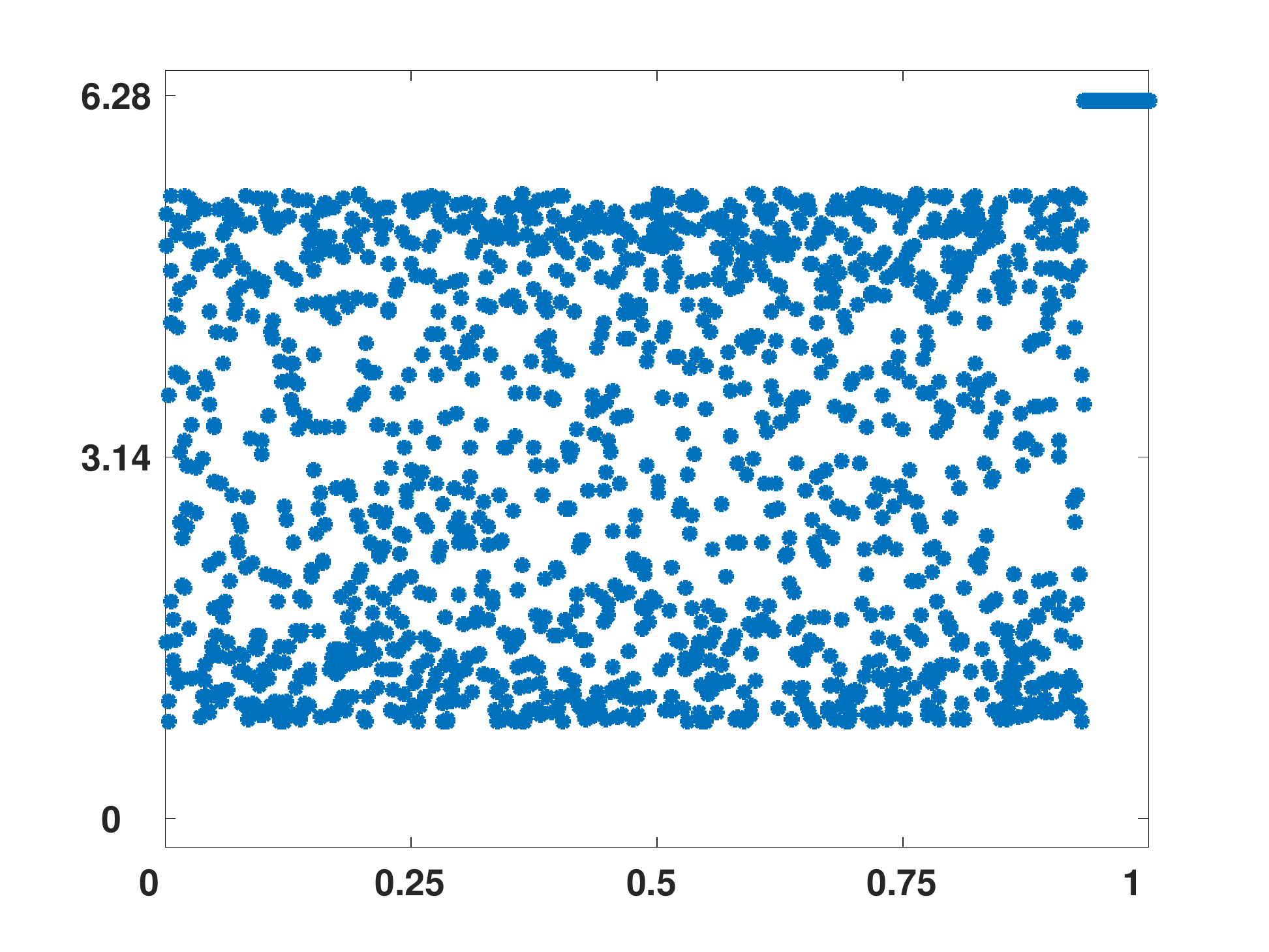}
\end{center}
\caption{The asymptotic states of \eqref{mKM}. The patterns in (\textbf{a}-\textbf{d})
correspond to the solutions of \eqref{mzdot} shown in the corresponding plots of
Fig.~\ref{f.2}. 
}
\lbl{f.3}
\end{figure}

Recalling $\mathcal{R}[z(0,\cdot)]>0,$ 
$\mathcal{R}[z(t,\cdot)]\searrow 0$ as $t\to\infty$, i.e., $z(t,\cdot)$ approaches an equilibrium from
$\tilde{\mathcal{E}}_2$. Next, we characterize the limiting state of the system. Since the initial condition is 
constant over each of the intervals $I^\pm$, so is the solution 
\be\lbl{zpm}
\begin{split}
z(x,t) &\equiv  z^-(t), \quad x\in I^-,\\
z(x,t) &\equiv  z^+(t), \quad x\in I^+,
\end{split}
\ee
Since $z^-(0)=-1+\delta$ and $z^-(t)\ge -1$, we have 
\be\lbl{tight-bound}
\left| z^-(t)+1\right| \le \delta,
\ee
i.e., the solution of the repulsively coupled KM \eqref{rKM} remains approximately synchronized over $I^-$.

Denote
$$
z^{\pm}_\infty := \lim_{t\to\infty} z^\pm (t).
$$
Further, since $\mathcal{R}[z(t,\cdot)]\to 0$ as $t\to\infty$,
we have
$$
z^+_\infty \int_{x_0}^1 y^{-\alpha} dy =-z^-_\infty \int^{x_0}_0 y^{-\alpha} dy.
$$
and
\be\lbl{balance}
z_\infty^+ = -z_\infty^- {x_0^{1-\alpha}\over 1-x_0^{1-\alpha}}.
\ee
The combination of \eqref{tight-bound} and \eqref{balance} yields
\be\lbl{zinfty-bound}
(1-\delta) {x_0^{1-\alpha}\over 1-x_0^{1-\alpha}}\le z_\infty^+\le {x_0^{1-\alpha}\over 1-x_0^{1-\alpha}}.
\ee
This double inequality combined with \eqref{tight-bound} yields tight estimates for the 
asymptotic state $z_\infty$ in $I^+$. Estimates \eqref{tight-bound} and \eqref{zinfty-bound}
characterize the asymptotic states for initial conditions $\mathcal{R}[z(0,\cdot)]>0$
(Fig.~\ref{f.2}\textbf{a},\textbf{b}). The complementary case $\mathcal{R}[z(0,\cdot)]>0$ is
analyzed similarly. 

\section{Discussion}\lbl{sec.discuss}
\setcounter{equation}{0}

The results of this study highlight the effects of the scale free connectivity for the 
dynamics of large networks. We found that  the synchronizability of the KM on sparse power law graphs 
is at least as good as it is on dense graphs. Moreover, the synchronization
threshold can be made arbitrarily low by varying the exponent of the power law 
degree distribution (cf.~\eqref{critical-pwl}). The imprint of the power law distribution is 
clearly seen in the stable chimera-like patterns generated by the repulsively coupled model
(Fig.~\ref{f.1} \textbf{d},\textbf{f}). Patterns shown in Fig.~\ref{f.1} and \ref{f.3} demonstrate a 
remarkable
ability of the attractors of the network to ``remember'' the initial condition on a continuous 
scale. Note that by continuously varying the parameter $x_0$ in the initial condition
$z^{(step)}_{\delta,x0}$ (i.e., by varying the the distribution of the positions of oscillators at time 
$0$), we are effectively changing the asymptotic state $z_\infty$ (Fig.~\ref{f.1} \textbf{a}-\textbf{c})
and, therefore, the asymptotic distribution of the oscillators (Fig.~\ref{f.1} \textbf{d}-\textbf{f}).
The memory of the initial conditions, which for the model at hand can be understood by studying the
reduced equation \eqref{zdot}, appears to be a common feature of nonlocally coupled networks
(see also \S 6.2 in \cite{Med14a}). In computational neuroscience, there has been a search 
for mechanisms  implementing continuous attractors in network models.
The models proposed, as a rule, in this context suffer from the lack of structural stability, i.e., the desired 
continuous attractor can be destroyed by small perturbations of parameters (see, e.g., \cite{Seu98}).
On the other hand, networks of nonlocally coupled oscillators, like the one presented in this
paper paper,  provide a robust mechanism for dependence of the attractor on initial conditions
on the continuous scale. Finally, the repulsively coupled KM on power law graphs \eqref{rAKM}
and its modification \eqref{mKM} provide a new simple mechanism for generating chimera 
states.

In conclusion, we note that the analysis in Section~\ref{sec.repulsive} is done for the averaged equation,
which approximates the original model on a random graph on finite time intervals. Therefore, the 
results for the averaged model reported in this section may hold for the original model only 
transitively. We believe that these results nonetheless give valuable insights into the dynamics
of coupled systems on power law graphs, and the asymptotic states of such systems should be 
investigated further.

\vskip 0.2cm
\noindent
{\bf Acknowledgements:}
This work was supported in part by the NSF DMS 1412066 (GM).

 \vfill\newpage
\bibliographystyle{amsplain}
\def\cprime{$'$} \def\cprime{$'$} \def\cprime{$'$}
\providecommand{\bysame}{\leavevmode\hbox to3em{\hrulefill}\thinspace}
\providecommand{\MR}{\relax\ifhmode\unskip\space\fi MR }
\providecommand{\MRhref}[2]{%
  \href{http://www.ams.org/mathscinet-getitem?mr=#1}{#2}
}
\providecommand{\href}[2]{#2}

\end{document}